\documentclass[11pt]{article}

\usepackage[square,authoryear]{natbib}
\usepackage{marsden_article}
\usepackage[all]{xy}

\newcommand{\newtext}[1]{#1}

\def\intprod{\vrule height 0pt depth 0.4pt width 3pt \vrule height 7pt depth 0.4pt
\kern 3pt}

\def\prodint{\kern 2pt \vrule height 7pt depth 0.4pt  \vrule %
height 0pt depth 0.4pt width 3pt %
\kern 2pt}

\begin{document}
\title{Multi-Dirac Structures and Hamilton-Pontryagin Principles for Lagrange-Dirac Field Theories}

\author{Joris Vankerschaver\thanks{Research supported by a Postdoctoral Fellowship of the Research Foundation--Flanders (FWO-Vlaanderen); 
email: joris.vankerschaver@gmail.com}
 \\Department of Mathematics\\
Ghent University\\ B-9000 Ghent, Belgium\\
\and
Hiroaki Yoshimura\thanks{Research partially supported by JSPS Grant-in-Aid
20560229, JST-CREST and Waseda University Grant for SR 2010A-606; email: yoshimura@waseda.jp}
\\ Applied Mechanics and Aerospace Engineering
\\ Waseda University
\\ Okubo, Shinjuku, Tokyo 
\\ 169-8555, Japan\\
 \and 
 Jerrold E. Marsden\thanks{Research partially supported by 
NSF grant DMS-0505711; email: marsden@cds.caltech.edu}
\\ Control and Dynamical Systems
\\ California Institute of Technology 107-81
\\ Pasadena, CA 91125
%
%\date{October 20, 2008}
%\date{May 2, 2009}
%\date{Sep 11, 2009}
%\date{Oct 26, 2009}
%\date{Nov. 11, 2009}
%\date{Nov. 25, 2009}
%\date{January 19, 2010}
%\date{January 21, 2010}
%\date{January 29, 2010}
%\date{February 25, 2010}
%\date{March 3, 2010}
%\date{June 3, 2010}
%\date{June 18, 2010}
%\date{July 8, 2010}
%\date{July 13, 2010}
%\date{July 19, 2010}
\date{July 28, 2010}
}
\maketitle

\begin{abstract}
\noindent
The purpose of this paper is to define the concept of multi-Dirac structures and to describe their role in the description of classical field theories.  We begin by outlining a variational principle for field theories, referred to as the Hamilton-Pontryagin principle, and we show that the resulting field equations are the Euler-Lagrange equations in implicit form.  Secondly, we introduce multi-Dirac structures as a graded analog of standard Dirac structures, and we show that the graph of a multisymplectic form determines a multi-Dirac structure.  We then discuss the role of multi-Dirac structures in field theory by showing that the implicit field equations obtained from the Hamilton-Pontryagin principle can be described intrinsically using multi-Dirac structures.  Furthermore, we show that any multi-Dirac structure naturally gives rise to a multi-Poisson bracket.  We treat the case of field theories with nonholonomic constraints, showing that the integrability of the constraints is equivalent to the integrability of the underlying multi-Dirac structure.   We finish with a number of illustrative examples, including time-dependent mechanics, nonlinear scalar fields and the electromagnetic field.
\end{abstract}

\tableofcontents

\section{Introduction}
 
To put this paper into context, we begin by giving a historical overview of the use of Dirac structures and the associated variational principles in classical mechanics. 

\paragraph{Dirac Structures and Lagrange-Dirac Systems in Mechanics.}

Originally developed by \cite{CoWe1988, Cour1990} and \cite{Dorfman1993}, 
Dirac structures unify (pre-)symplectic and Poisson structures and find their inspiration in Dirac's theory of constraints, which was established by \cite{Dirac1950} for constrained systems including constraints due to degeneracy of Lagrangians.

It was soon realized that Dirac structures play an important role in mechanics.  Indeed, a wide class of {\it implicit Hamiltonian systems} can be naturally described in terms of Dirac structures.  In this paper, we will refer to these systems as {\it Hamilton-Dirac systems}. In particular, it was shown that interconnected systems, such as LC circuits, and   nonholonomic systems can be effectively formulated in the context of implicit Hamiltonian systems (\cite{Cour1990, Dorfman1993, CoWe1988,VaMa1995}).  On the other hand, it was shown by \cite{YoMa2006a} that Dirac structures induced from distributions on configuration manifolds naturally yield a notion of {\it implicit Lagrangian systems} or {\it Lagrange-Dirac dynamical systems}, which allow for the description of mechanical systems with degenerate Lagrangians and with nontrivial constraint distributions.   Finally, as shown by \cite{LeOh2008}, there also exists a natural \emph{discrete} definition of Dirac structures, which can be used to derive discrete, geometric approximations of the equations of motion for mechanical systems.  From the viewpoint of mechanical systems with symmetry, reduction theory of Dirac dynamical systems was developed by \cite{Va1998, BL2000} and \cite{BR2004}, while reduction theory of Lagrange-Dirac systems was developed by \cite{YoMa2007, YoMa2009} and \cite{CMRY2010}.  In this paper, we will define a graded analogue of Dirac structures, referred to as {\bfi multi-Dirac structures}, and we will show that they can be used to describe the implicit Euler-Lagrange equations of first-order field theories.

\paragraph{The Hamilton-Pontryagin Variational Principle.}

Let $L(q, v)$ be a Lagrangian.   The {\it Hamilton-Pontryagin principle}, originally coined by \cite{Livens1919}, is a variational principle in which both the position coordinates $q$ and the velocity coordinates $v$ are varied independently.  The relation $\dot{q} = v$ is then imposed by means of a Lagrange multiplier $p$, leading to an action functional of the form 
\begin{equation} \label{HPmech}
	S(q, v, p) = \int_{t_0}^{t_1} 
		\Big( 
	L(q(t), v(t)) + \left< p(t), \dot{q}(t) - v(t) \right> \Big) dt.  
\end{equation}
By taking arbitrary variations with respect to $q$, $v$ and $p$, we obtain the Euler-Lagrange equations in implicit form: 
\[
	p_i = \frac{\partial L}{\partial v^i}, \quad 
	\dot{p}_i = \frac{\partial L}{\partial q^i}, 
	\quad \text{and} \quad \dot{q}^i = v^i.
\]
In the context of variational principles, it was shown by \cite{YoMa2006b} that the Hamilton-Pontryagin principle and the associated implicit equations of motion can be described in an intrinsic way by means of Lagrange-Dirac structures.  Similar variational principles as the Hamilton-Pontryagin principle have been used by \cite{BoRaMa2009} for the design of accurate variational integrators.  The Hamilton-Pontryagin principle also appears under a different guise in the mechanics of continuous media, where it is related to the \emph{Hu-Washizu variational principle} (see \cite{Wa1968}).  We will extend the Hamilton-Pontryagin principle to deal with first-order field theories.

\paragraph{Stokes-Dirac Structures and Infinite-Dimensional Hamiltonian Systems.}

While there is as of yet no comprehensive theory of Dirac structures for classical field theories, an important contribution towards this goal was made by  \cite{VaMa2002}, who introduced the concept of {\it Stokes-Dirac structure}, an infinite-dimensional Dirac structure associated with the exterior derivative and Stokes's theorem.  This structure can be used for describing Hamilton's equations for a class of field theories with boundary energy flow.  Recently, it has been clarified by \cite{VaYoMa2010} that the Stokes-Dirac structure can be obtained through symmetry reduction of  a canonical Dirac structure that is induced by the canonical Poisson structure on the cotangent bundle.  In contrast to Stokes-Dirac structures, the multi-Dirac structures in this paper are covariant in the sense that time and space are treated on an equal footing.  Hence, we expect that a multi-Dirac structure induces a Stokes-Dirac structure after choosing a 3+1 decomposition of spacetime; we will return to this issue at the end of the paper.

\paragraph{Multisymplectic Structures.}

Classical field theories can be described using the multi-symplectic formalism, where the jet bundle and its dual of a configuration bundle play the role of respectively the tangent and cotangent bundle in mechanics.  The dual of a jet bundle can be equipped with a canonical multisymplectic form, which is a higher-degree analogue of the canonical symplectic form on a cotangent bundle.  By means of the Legendre transformation, this form can be pulled back to the jet bundle.  These multisymplectic forms provide a natural geometric setting for the Lagrangian and Hamiltonian description of field theories.  As the literature on multisymplectic field theories is very extensive, we can only refer to a subset of it.  For fundamental aspects, see \cite{GolSter73,KiTu1979, BiSnFi1988,CaCrIb1991, Sardanashvily1995, CIdeLe96, GIMM1997, GIMM1999, Brid1997, EcMuRo2000} and also the references therein. 

\paragraph{Nonholonomic Mechanics and  Field Theories.}
Much effort has been dedicated to mechanical systems with nonholonomic constraints (see, for instance, \cite{VF1981,BaSn1993, Bloch2003,CeMaRa2001}). Amongst them,  it was shown by \cite{BKMM1996} that the Lagrange-d'Alembert principle plays an essential role in formulating the equations of motion from the Lagrangian viewpoint of symmetry and reduction.  On the Hamiltonian side, it was shown by \cite{VaMa1994} that constrained Hamiltonian systems can be developed from the viewpoint of \emph{almost-Poisson structures}, which satisfy the Jacobi identity only if the underlying constraints are integrable.  This was later extended to the Lagrangian context by \cite{CaLeMa1999}.  Further, a notion of implicit Hamiltonian systems was developed by \cite{VaMa1995} and \cite{Va1998} in the context of Dirac structures (see also \cite{BL2000}). 
Nonconservative systems with external forces appearing in servomechanisms were also illustrated in the context of the constrained Hamiltonian systems by \cite{Marle1998}. The equivalence of the Lagrangian and Hamiltonian formalisms for nonholonomic mechanical systems was demonstrated by \cite{KM1997,KM1998} together with their intrinsic expressions. 

By analogy with nonholonomic mechanics, nonholonomic classical field theories have been developed; for instance, \cite{BdeLMDS02} proposed a generalized version of the d'Alembert principle by using the so-called Chetaev condition to represent the bundles of admissible reaction forces associated to given nonlinear nonholonomic constraints in the context of the multisymplectic formalism for classical field theories. \cite{VCdeLMD2005} extended this by relaxing the Chetaev-type conditions to allow for constraint forms that need not be determined by the constraint forces themselves.  A physical example of a nonholonomic field theory appears in \cite{Vank2007}.  For more information about nonholonomic field theories, see also \cite{KrupkovaVolny2}.

Just as (almost) Dirac structures and almost-Poisson brackets play an important role in the description of nonholonomic mechanics, so too can nonholonomic field theories be described in terms of a multi-Dirac structure induced by the multisymplectic form and the constraint distribution.   This will be yet another class of multi-Dirac structures and as in the case of mechanics we will show that these multi-Dirac structures are integrable if and only if the underlying constraints are integrable.

\newtext{
As an interesting aside, we will show that the specification of a nonholonomic multi-Dirac structure is not sufficient to completely determine the nonholonomic field equations, and that an extra assumption needs to be introduced.  This assumption is trivially true in the case of mechanical systems with nonholonomic constraints, and illustrates the difference between classical mechanics and field theory with nonholonomic constraints.  An overview of other conceptual differences can be found in \cite{Krupkova05, VCdeLMD2005}.
}

\paragraph{Graded Poisson Structures and Higher-Order Dirac Structures.}

Not only do multi-Dirac structures appear naturally as phase spaces for implicit first-order field theories, they also have a number of mathematical properties that are interesting in their own right.  As a special case, multi-Dirac structures include the standard Dirac structures of \cite{Cour1990} and the higher-order Dirac structures of \cite{Zambon2010}.  We will also show that to each $k$-form $\Omega$, we may associate a multi-Dirac structure $D$, which is (in a sense defined below) the graph of $\Omega$.  The condition that $D$ is integrable is then equivalent to $\mathbf{d} \Omega = 0$.  Furthermore, just as a standard Dirac structure induces a Poisson bracket on a restricted class of functions, we will show that a multi-Dirac structure endows a subset of the space of forms $\Omega^\ast(M)$ with a graded Poisson bracket, where the graded Jacobi identity is satisfied up to exact forms.  Similar brackets were defined in the multisymplectic literature (see section~\ref{sec:multidirac} for an overview) and we expect that multi-Dirac structures will be a useful tool in studying their properties.  We also mention that a similar but different graded structure was introduced by \cite{Bridges05, BrHyLa2010} in their study of first-order field theories.

\paragraph{Layout of the Paper.}

In \S2, we briefly review the geometry of jet bundles as basic mathematical ingredients and especially describe the geometry of the Pontryagin bundle $M=J^{1}Y \oplus Z$ over a bundle $\pi_{XY}: Y \to X$ for {\it Lagrange-Dirac field theories}. In \S3, we develop {\it implicit Lagrangian field theories} in the context of the {\it Hamilton-Pontryagin principle} on $M$, which naturally yield {\it implicit Euler-Lagrangian field equations}. This is a natural extension of the variational formulation of the Hamilton-Pontryagin principle for mechanical systems as in \cite{YoMa2006b}. In \S4, we develop multi-Dirac structures on $M$ by introducing a natural pairing between multivectors and forms and we show that the graph of a multisymplectic form (in some suitable sense) yields a multi-Dirac structure.  

Further, it is shown in \S5 that the implicit Euler-Lagrange equations for field theories can be described using multi-Dirac structures; this leads us to the definition of Lagrange-Dirac systems.  In \S6, we also explore the {\it induced multi-Dirac structure from nonholonomic distributions} for the case of affine constraints and the associated {\it nonholonomic Lagrange-Dirac systems for field theories}. Furthermore, we develop a constrained version of the Hamilton-Pontryagin principle, referred to as the {\it Lagrange-d'Alembert-Pontryagin principle} to formulate {\it implicit Lagrange-d'Alembert field equations}, which are equivalent with the nonholonomic Lagrange-Dirac systems for field theories. In \S7, we demonstrate the present theory by examples of nonlinear Klein-Gordon scalar fields, electromagnetic fields as well as time-dependent mechanical systems with affine constraints. Finally, conclusions are given in \S8 together with some remarks for future works.   

\paragraph{Acknowledgements.}  

\newtext{
We are very grateful to Thomas Bridges, Henrique Bursztyn, Frans Cantrijn, Marco Castrill{\'o}n-L{\'o}pez, Mark Gotay, Juan-Carlos Marrero, David Mart\'{\i}n de Diego, Chris Rogers and Marco Zambon who kindly provided several very useful remarks.  In particular, we thank Marco Zambon for pointing out a flaw in Proposition~\ref{prop: dirac-const} in an earlier version of this paper.
}

\section{The Geometry of Jet Bundles}

\label{sec:geometry}

In this section, we provide a quick overview of geometry of jet bundles for the treatment of Lagrangian field theories.  Most of the material in this section is standard, and can be found in \cite{Saunders89, GIMM1997} and the references therein.

\paragraph{Jet Bundles.}

Let $X$ be an oriented manifold with volume form $\eta$, which in many examples is spacetime, and let $\pi_{XY}: Y \to X$ be a finite-dimensional fiber bundle which we call the {\bfi covariant configuration bundle}. The physical fields are sections of this bundle, which is the covariant analogue of the configuration space in classical mechanics. For future reference, we suppose that the dimension of $X$ is $n + 1$ and that $\pi_{XY}$ is a bundle of rank $N$, so that $\dim Y = n + N + 1$.  Coordinates on $X$ are denoted $x^{\mu},\, \mu=1,2,...,n+1$, and fiber coordinates $Y$ are denoted by $y^{A}, \, A=1,...,N$ such that a section $\phi: X \rightarrow Y$ of $\pi_{XY}$ has coordinate representation $\phi(x)=(x^{\mu},y^{A}(x))=(x^{\mu},y^{A})$. 

The analogue in classical field theory of the tangent bundle in mechanics is the {\bfi first jet bundle} $J^1 Y$, which consists of equivalence classes of local sections of $\pi_{XY}$, where we say that two local sections $\phi_{1}$, $\phi_{2}$ of $Y$ are equivalent at $x \in X$ if their Taylor expansions around $x$ agree to the first order.  In other words, $\phi_1$ and $\phi_2$ are equivalent if $\phi_1(x) = \phi_2(x)$ and $T_x \phi_1 = T_x \phi_2$.  It follows that an equivalence class $[\phi]$ of local sections can be identified with a linear map $\gamma : T_x X \rightarrow T_y Y$ such that $T \pi_{XY} \circ \gamma = \mathrm{Id}_{T_x X}$.  As a consequence, $J^1 Y$ is a fiber bundle over $Y$, where the projection $\pi_{J^1 Y, Y}: J^1 Y \rightarrow Y$ is defined as follows: let $\gamma: T_x X \rightarrow T_y Y$ be an element of $J^1 Y$, then $\pi_{J^1 Y, Y}(\gamma) := y$.  Coordinates on $J^{1}Y$ are denoted $(x^{\mu},y^{A},v^{A}_{\mu})$, where the fiber coordinates $v^{A}_{\mu}$ represent the first-order derivatives of a section.  They are defined by noting that any $\gamma \in J^1 Y$ is locally of the form
\[
	\gamma =d x^\mu \otimes \left( \frac{\partial}{\partial x^\mu} + v^A_\mu \frac{\partial}{\partial y^A} \right).
\]
The first jet bundle $J^1 Y$ is an \emph{affine} bundle over $Y$, with underlying vector bundle the bundle $L(TX, VY)$ of linear maps from the tangent space $TX$ into the vertical bundle $VY$, defined as  
\[
V_{y}Y=\{ v\in T_{y}Y \mid T\pi_{XY}(v)=0 \}, \quad \text{for $y \in Y$}.
\]

Given any section $\phi: X \to Y$ of $\pi_{XY}$, its tangent map $T_{x}\phi$ at $x \in X$ is an element of $J^{1}_yY$, where $y = {\phi(x)}$.   Thus, the map $x \mapsto T_{x}\phi$ defines a section of  $J^{1}Y$, where now $J^1 Y$ is regarded as a bundle over $X$. This section is denoted $j^{1}\phi$ and is called the {\bfi first jet prolongation} of $\phi$. In coordinates, $j^{1}\phi$ is given by 
\[
x^{\mu} \mapsto (x^{\mu},y^{A}(x^{\mu}),\partial_{\nu}y^{A}(x^{\mu})),
\]
where $\partial_{\nu}=\partial/\partial{x^{\nu}}$. A section of the bundle $J^{1}Y \to X$ which is the first jet prolongation of a section $\phi: X \to Y$ is said to be {\bfi holonomic}.

\paragraph{Dual Jet Bundles.} 

Next, we consider the field-theoretic analogue of the cotangent bundle. We define the {\bfi dual jet bundle} $J^{1}Y^{\star}$ to be the vector bundle over $Y$ whose fiber at $y \in Y_{x}$ is the set of affine maps from $J^{1}_{y}Y$ to $\Lambda^{n+1}_{x}X$, where $\Lambda^{n+1} X$ denotes the bundle of $(n+1)$-forms on $X$. Note that since the space of affine maps from an affine space into a vector space forms a vector space, $J^{1}Y^{\star}$ is a vector bundle despite the fact that $J^{1}Y$ is only an affine bundle. A smooth section of $J^{1}Y^{\star}$ is therefore an affine bundle map of $J^{1}Y$ to $\Lambda^{n+1}X$.  Any affine map from $J^{1}_{y}Y$ to $\Lambda^{n+1}_{x}X$ can locally be written as 
\[
v^{A}_{\mu} \mapsto (p+p_{A}^{\mu}v^{A}_{\mu})\,d^{n+1}x,
\]
where $d^{n+1}x := dx^{1}\wedge dx^{2} \wedge \cdots \wedge dx^{n+1}$ is a coordinate representation of the volume form $\eta$, so that coordinates on $J^{1}Y^{\star}$ are given by $(x^\mu, y^A,p_{A}^{\mu}, p)$.

Throughout this paper we will employ another useful description of $J^{1}Y^{\star}$. Consider again the bundle $\Lambda^{n+1}Y$ of $(n+1)$-forms on $Y$ and let $Z \subset \Lambda^{n+1}Y$ be the subbundle whose fiber over $y \in Y$ is given by
\[
Z_{y}=\{z \in \Lambda_{y}^{n+1} Y \mid \mathbf{i}_{v}\mathbf{i}_{w}z=0 \;\text{for all}\; v,w \in V_{y}Y \},
\]
where $\mathbf{i}_{v}$ denotes left interior multiplication by $v$.  In other words, the elements of $Z$ vanish when contracted with two or more vertical vectors.  The bundle $Z$ is canonically isomorphic as a vector bundle over $Y$ to $J^1 Y^\star$; this can be easily understood from the fact that the elements of $Z$ can locally be written as 
\[
z=p\,d^{n+1}x +p_{A}^{\mu}dy^{A} \wedge d^{n}x_{\mu},
\]
where $d^{n}x_{\mu} :=\partial_{\mu}$ {\Large$\lrcorner$} $d^{n+1}x$.

From now on, we will silently identify $J^1 Y^\star$ with $Z$.  The duality pairing between an element $\gamma \in J^1_{y_x} Y$ and $z \in Z_{y_x}$ can then be written as 
\[
	\left< \gamma, z \right> = \gamma^\ast z \in \Lambda_x^{n+1} X.
\]
In coordinates, $\left< \gamma, z \right> = (p_A^\mu v^A_\mu + p ) d^{n+1} x$.

\paragraph{Canonical Multisymplectic Forms.}

Analogous to the canonical symplectic forms on a cotangent bundle, there are canonical forms on $J^{1}Y^{\star}$.  
Let us first define the {\bfi canonical $(n+1)$-form} $\Theta$ on $J^1 Y^\star \cong Z$ by
\begin{equation} \label{canonnform}
\begin{split}
\Theta(z)(u_{1},...,u_{n+1})&=z(T\pi_{Y, J^1Y^\star}( u_{1}),...,T\pi_{Y, J^1Y^\star} (u_{n+1}))\\
&=(\pi_{Y, J^1Y^\star}^{\ast}z)(u_{1},...,u_{n+1}),
\end{split}
\end{equation}
where we have interpreted $z \in Z \cong J^1 Y^\star$ as before as an $(n+1)$-form on $Y$, and $u_{1},...,u_{n+1} \in T_{z} Z$. The {\bfi canonical multisymplectic  $(n+2)$-form} $\Omega$ on $J^1Y^\star$ is now defined as
\[
\Omega = -\mathbf{d}\Theta.
\]

Denoting again $d^{n}x_{\mu} := \partial_{\mu}$ {\Large$\lrcorner$} $d^{n+1}x$, one has the following coordinate expression for $\Theta$:
\[
\Theta=p_{A}^{\mu}dy^{A} \wedge d^{n}x_{\mu}+p\,d^{n+1}x,
\]
while $\Omega$ is locally given by 
\begin{equation} \label{msform}
\Omega=dy^{A} \wedge dp_{A}^{\mu}\wedge d^{n}x_{\mu}-dp \wedge d^{n+1}x.
\end{equation}
It is easy to show (see \cite{CIdeLe99}) that $\Omega$ is non-degenerate in the sense that $\mathbf{i}_{\mathcal{X}}\Omega=0$ implies that $\mathcal{X}=0$.  Moreover, $\Omega$ is trivially closed: $\mathbf{d} \Omega = 0$.  Forms which are both closed and non-degenerate are referred to as {\bfi multisymplectic} forms.  If the degree of $\Omega$ is two, $\Omega$ is just a symplectic form, while if $\Omega$ is of maximal degree, $\Omega$ is a volume form.  More examples of multisymplectic manifolds can be found in \cite{CIdeLe99}.

A form which is non-degenerate but not necessarily closed is also referred to as an {\bfi almost multisymplectic} form.  On the other hand, in the context of field theory, we will often refer to a closed but possibly degenerate form as a {\bfi pre-multisymplectic} form.

\paragraph{Lagrangian Densities and the Covariant Legendre Transformation.} 

A {\bfi Lagrangian density} is a smooth map $\mathcal{L}=L\eta:J^{1}Y \to \Lambda^{n+1}X$.  In local coordinates, we may write 
\[
\mathcal{L(\gamma)}=L(x^{\mu},y^{A},v^{A}_{\mu})d^{n+1}x,
\]
where $L$ is a function on $J^1 Y$ to which we also refer as the Lagrangian.

The corresponding {\bfi covariant Legendre transformation} for a given Lagrangian density $\mathcal{L}:J^{1}Y \to \Lambda^{n+1}X$ is a fiber preserving map $\mathbb{F}\mathcal{L}:J^{1}Y \to J^{1}Y^{\star}$ over $Y$, which is given by the first order vertical Taylor approximation to $\mathcal{L}$:
\[
\left<\mathbb{F}\mathcal{L}(\gamma),\gamma^{\prime}\right>=\mathcal{L}(\gamma) + \frac{d}{d \epsilon}\bigg|_{\epsilon=0} \mathcal{L}(\gamma+ \epsilon (\gamma-\gamma^{\prime})),
\]
where $\gamma, \gamma^{\prime} \in J^{1}Y$. In coordinates, $\mathbb{F}\mathcal{L}(x^\mu, y^A, v^A_\mu) = (x^\mu, y^A, p_A^\mu, p)$, where
\begin{equation}\label{CovLegTrans}
p_{A}^{\mu}=\frac{\partial{L}}{\partial{v^{A}_{\mu}}}, \qquad p=L-\frac{\partial{L}}{\partial{v^{A}_{\mu}}}v^{A}_{\mu}.
\end{equation}

\paragraph{Multivector Fields.}

A useful tool in the description of classical field theory is provided by multivector fields.  The general theory of multivector fields can be found in \cite{Tu1974}, \cite{Marle1997} and the references therein, while applications of multivector calculus to classical field theory can be found in (for instance) \cite{EcMuRo2002} and \cite{FoPaRo2005}.

For the sake of generality, we will introduce multivector fields first on an arbitrary manifold $P$.  Later on, we will specialize to the case where $P$ is a jet bundle, its dual or a combination of both.  We denote by $T^r P$ the $r$-fold tangent bundle, that is, the $r$-fold exterior power $\Lambda^r(TP)$ of $TP$ with itself.  A {\bfi multivector field} of degree $r$ on a manifold $P$ is a section $\mathcal{X}_r$ of the $r$-fold tangent bundle $T^r P$.  We say that $\mathcal{X}_r$ is {\bfi decomposable} when there exist $r$ vector fields $X_1, \ldots, X_r$ on $P$ such that $\mathcal{X}_r = X_1 \wedge \cdots \wedge X_r$.  An {\bfi integral manifold} of an $r$-multivector field $\mathcal{X}_r$ is an embedded submanifold $S \hookrightarrow P$ of dimension $r$ such that $\mathcal{X}_r$ spans $T^r S$ at every point.  Note that the existence of integral manifolds is not automatically guaranteed but depends on certain integrability conditions, which we now describe.

The integral manifolds of a decomposable $r$-multivector field $\mathcal{X}_r  = X_1 \wedge \cdots \wedge X_r$ can be described in an alternative way by noting that $\mathcal{X}_r$ induces a distribution $\Delta_{\mathcal{X}_r}$ of rank $r$, which is locally spanned by the $r$ vector fields $X_1, \ldots, X_r$.  The integral manifolds of $\mathcal{X}_r$ then coincide with the integral manifolds of $\Delta_{\mathcal{X}_r}$, so that in this case a necessary and sufficient condition for the existence of integral manifolds is given by the integrability of $\Delta_{\mathcal{X}_r}$.

Assume now that $P$ is the total space of a fiber bundle $\pi_{XP}: P \rightarrow X$, where $X$ is equipped with a volume form $\eta$ and $\dim X = n + 1$.  As mentioned previously, most of the times $P$ will be a jet bundle, its dual, or a combination of both.  For fiber bundles, more can be said about multivector fields and their integral manifolds. We will often consider $(n + 1)$-multivector fields on the total space $P$ which are decomposable and satisfy the following {\bfi normalization condition}:
\[
	\mathbf{i}_{\mathcal{X}_{n+1}} (\pi_{XP}^\ast \eta) = 1.
\] 
It can be shown that the integral manifolds of such a multivector field $\mathcal{X}_{n+1}$ are given as the image of a section $\psi: X \rightarrow P$.  This is clear from the coordinate representation of $\mathcal{X}_{n+1}$, as we show.  Let $X$ have coordinates $(x^\mu)$, $\mu = 1, \ldots, n+1 $, and choose a system of bundle coordinates $(x^\mu, u^A)$, $A = 1, \ldots, N$ on $P$.  Then, the multivector field $\mathcal{X}_{n + 1}$ can locally be written as 
\[
	\mathcal{X}_{n+1} = \bigwedge_{\mu = 1}^{n + 1}
		\left( \frac{\partial}{\partial x^\mu} + C^A_\mu(x, u) \frac{\partial}{\partial u^A} \right),
\]
where the coefficient functions $C^A_\mu(x, u)$ are local functions on $P$.  Now, we consider a section $\phi : X \rightarrow P$ with local coordinate representation $x^\mu \mapsto (x^\mu, \phi^A(x))$.  The section $\phi$ determines an integral manifold of $\mathcal{X}_{n + 1}$ if it satisfies the following system of PDEs:
\[
	\frac{d \phi^A}{d x^\mu} = C^A_\mu(x^\nu, \phi^b(x)).
\]

\section{The Hamilton-Pontryagin Principle for Field Theories}

In this section, we introduce the {\it Hamilton-Pontryagin action principle for classical field theories}.  Let us first introduce the Pontryagin bundle 
\[
M :=J^{1}Y \times_Y Z,
\]
as the fibered product over $Y$ of the jet bundle $J^1 Y$ and its dual $Z$.  This bundle plays a similar role as the Pontryagin bundle $TQ \oplus T^{\ast}Q$ over a configuration manifold $Q$ for the case of classical mechanics.  Note also that $M$ is a bundle over $Y$ and also over $X$.

\paragraph{The Generalized Energy Density.} 

Let $\mathcal{L}: J^{1}Y \to \Lambda^{n+1}X$ be a Lagrangian density which is possibly degenerate. We define the {\bfi generalized energy density} associated to $\mathcal{L}$ by the map  $\mathcal{E}:M \to \Lambda^{n+1}X$ defined as 
\[
\mathcal{E}(\gamma,z):=E(\gamma,z)\eta=\left< z, \gamma \right>-\mathcal{L}(\gamma)
\]
for $(\gamma,z) \in M$.  In local coordinates $(x^{\mu},y^{A},v^{A}_{\mu},p^{\mu}_{A}, p)$ on $M$, the generalized energy density is represented by  
\begin{equation} \label{generalizedenergy}
\begin{split}
\mathcal{E}&=E(x^{\mu},y^{A},v^{A}_{\mu},p,p^{\mu}_{A})d^{n+1}x\\
&=\left(p+p_{A}^{\mu}v^{A}_{\mu}-L(x^{\mu},y^{A},v^{A}_{\mu})\right)d^{n+1}x,
\end{split}
\end{equation}
where $E=p+p_{A}^{\mu}v^{A}_{\mu}-L(x^{\mu},y^{A},v^{A}_{\mu})$ is called the {\bfi generalized energy} on $M$.

\paragraph{Pre-Multisymplectic Forms on $M$.} 

Recall the definition of the canonical forms $\Theta$ and $\Omega:=-\mathbf{d}\Theta$ on the dual jet bundle $Z \cong J^1 Y^\star$, and let $\pi_{Z M}: M \rightarrow Z$ be the projection onto the second factor.  The forms $\Theta$ and $\Omega$ can be pulled back along $\pi_{ZM}$ to yield corresponding forms on $M$:
\[
	\Theta_M := \pi_{ZM}^\ast \Theta \quad \text{and} \quad 
	\Omega_M := \pi_{ZM}^\ast \Omega.
\]
Note that $\Omega_M$ cannot be multisymplectic since it has a non-trivial kernel: for all $v \in T J^1 Y$ we have that $\mathbf{i}_v \Omega_M = 0$.  As a result, we will refer to $\Omega_M$ as the {\bfi canonical pre-multisymplectic $(n + 2)$-form} on $M$.  Whenever there is no possibility of confusion, we will omit the subscript `$M$', and denote both canonical forms on $M$ respectively by $\Theta$ and $\Omega$.   Finally, for any Lagrangian density $\mathcal{L}$ with associated energy density $\mathcal{E}$, we introduce another pre-multisymplectic $(n + 2)$-form on $M$ by
\[
\Omega_{\mathcal{E}}=\Omega_{M}+\mathbf{d}\mathcal{E}.
\]

\paragraph{The Hamilton-Pontryagin Principle.}

Using the pre-multisymplectic forms $\Omega_M$ and $\Omega_{\mathcal{E}}$, we establish a Hamilton-Pontryagin variational principle for classical field theories which is similar to the expression \eqref{HPmech} for mechanical systems outlined in the introduction.

\begin{definition}\label{def:HPActionFunct}
Consider a Lagrangian density $\mathcal{L}$ with associated energy density $\mathcal{E}$.  The Hamilton-Pontryagin action functional is defined as 
\begin{equation}\label{HPActionFunc}
	S(\psi) = \int_{X}  \psi^{\ast}(\Theta_{M} -\mathcal{E}),
\end{equation}
where $\psi$ is a section of $\pi_{XM}: M \to X$.  In other words, $\psi$ can be written as $\psi=(\gamma, z)$ where $\gamma$ is a section of $\pi_{X,J^{1}Y}:J^{1}Y \to X$ and $z$ is a section of $\pi_{XZ}: Z \to X$.
\end{definition}

We have defined the Hamilton-Pontryagin principle in terms of the canonical form $\Theta_M$.  At first sight, this looks very different from the corresponding principle \eqref{HPmech} for mechanics.  However, by using the definition \eqref{canonnform}, we can rewrite the Hamilton-Pontryagin action functional as 
\[
	S(\psi) = \int_X \left( \left<z, j^1 \phi \right> 
		- \left<z, \gamma \right> + \mathcal{L}(\gamma) 
			\right),
\]
where we have written the section $\psi$ as $\psi = (\gamma, z)$.

A {\bfi vertical variation} of a section $\psi = (\gamma, z)$ of $\pi_{XM}: M \to X$ is a one-parameter family of diffeomorphisms $\eta_\lambda : M \rightarrow M$ such that $\eta_0$ is the identity, $\eta_\lambda$ is the identity outside of a compact subset $U$ of $M$, and $\eta_\lambda$ preserves the fibration $\pi_{XM} : M \rightarrow X$, or in other words $\pi_{XM} \circ \eta_\lambda = \pi_{XM}$.  As a result, the composition $\eta_\lambda \circ \psi$ is a one-parameter family of sections of $M$.
 
On the infinitesimal level, an infinitesimal variation of $\psi = (\gamma, z)$ is a $\pi_{XM}$-vertical vector field $\mathcal{V}_{M}: \mathrm{Im}\, \psi \to VM$, defined on the image of $\psi$ in $M$, and where $VM$ is the vertical subbundle of $TM$ defined by
\[
VM=\left\{ w \in TM \mid T\pi_{XM} (w) =0\right\}.
\]
Let us denote by $\mathfrak{X}^{\mathrm{V}}(M)$ the module of the vertical vector fields on $M$.

One says that $\psi=(\gamma, z)$ is a {\bfi critical point} of the action \eqref{HPActionFunc} if  $\delta S(\psi) = 0$ for all variations of $\psi$, where 
\[
\begin{split}
\delta S(\psi)&=\frac{d}{d\lambda}\bigg|_{\lambda=0}S(\psi_{\lambda})\\
&=\frac{d}{d\lambda}\bigg|_{\lambda=0} \int_{U} \psi_{\lambda}^{\ast} ( \Theta_{M} -\mathcal{E}),
\end{split}
\]
where $U$ is an open subset of $X$ with compact closure.

\begin{proposition}
A section $\psi = (\gamma, z)$ of $M$ is a critical point of the Hamilton-Pontryagin action functional if $\psi$ satisfies the 
{\bfi implicit Euler-Lagrange equations}
\begin{equation}\label{ELFieldEqn-M}
\begin{split}
\psi^{\ast}(\mathcal{V}_{M} \!\text{\Large$\lrcorner$} \;\Omega_{\mathcal{E}})=0,  \quad \mbox{for any \;$\mathcal{V}_{M} \in \mathfrak{X}^{\mathrm{V}}(M)$}.
\end{split}
\end{equation}
\end{proposition}

\begin{proof}
Recall that a variation of $\psi$ is given by $\psi_{\lambda}=\eta_{\lambda} \circ \psi$ where $\eta_{\lambda}: M \to M$ is the flow of a vertical vector field $\mathcal{V}_{M}$. Then, it follows that
\begin{equation*}
\begin{split}
\frac{d}{d\lambda}\bigg|_{\lambda=0}S(\psi_{\lambda})&=\frac{d}{d\lambda}\bigg|_{\lambda=0}\int_{U} \psi^{\ast}_{\lambda} (\Theta_{M} - \mathcal{E})\\
&=\int_{U} \psi^{\ast} \pounds_{{\mathcal{V}}_{M}} (\Theta_{M} -\mathcal{E})\\
&=-\int_{U} \psi^{\ast} (\mathcal{V}_{M} \! \text{\Large$\lrcorner$} \; (\Omega_{M} + \mathbf{d} \mathcal{E})) 
+ \int_U  \mathbf{d}( \psi^\ast (\mathcal{V}_M\! \text{\Large$\lrcorner$} \;
(\Theta_M - \mathcal{E})))
\\
&=-\int_{U} \psi^{\ast} (\mathcal{V}_{M} \! \text{\Large$\lrcorner$} \;\Omega_{\mathcal{E}})
+ \int_{\partial U}  \psi^\ast (\mathcal{V}_M\! \text{\Large$\lrcorner$} \;
(\Theta_M - \mathcal{E}))
\\
&=-\int_{U} \psi^{\ast} (\mathcal{V}_{M} \! \text{\Large$\lrcorner$} \;\Omega_{\mathcal{E}})
\end{split}
\end{equation*}
for all $\mathcal{V}_{M}$,
where we utilized Stokes' theorem and the fact that the $\pi_{XM}$-vertical vector field $\mathcal{V}_{M}: M \to VM$  is compactly supported in $U \subset X$.  A standard argument then shows that $\psi$ is a critical point of $S$ if and only if \eqref{ELFieldEqn-M} holds.
\end{proof}

\paragraph{Coordinate Expressions.} 
Employing local coordinates $(x^{\mu},y^{A},v^{A}_{\mu},p,p_{A}^{\mu})$ on $M$, the action functional is denoted by
\begin{equation} \label{coordhpprinciple}
\begin{split}
S(\gamma)&=\int_{U}  \psi^{\ast}(\Theta_{M} -\mathcal{E})\\
&=\int_{U} \left(p+p_{A}^{\mu} \frac{\partial y^{A}}{\partial x^{\mu}} \right)d^{n+1}x- \left\{\left(p+p_{A}^{\mu}v^{A}_{\mu})d^{n+1}x-L(x^{\mu},y^{A},v^{A}_{\mu}\right)d^{n+1}x \right\}\\
&=\int_{U} \left\{ p_{A}^{\mu} \left( \frac{\partial y^{A}}{\partial x^{\mu}}- v^{A}_{\mu} \right)+ L(x^{\mu},y^{A},v^{A}_{\mu})\right\} d^{n+1}x.
\end{split}
\end{equation}

\begin{proposition}
The Hamilton-Pontryagin principle induces the {\bfi local implicit Euler-Lagrange equations} for equation \eqref{ELFieldEqn-M}, which are given in coordinates by
\begin{equation}\label{local-ELFieldEqn}
\frac{\partial y^{A}}{\partial x^{\mu}}=v^{A}_{\mu}, \quad \frac{\partial p_{A}^{\mu}}{\partial x^{\mu}}= \frac{\partial L}{\partial y^{A}}, \quad p_{A}^{\mu}= \frac{\partial L}{\partial v^{A}_{\mu}}.
\end{equation}
\end{proposition}
\begin{proof}
By direct computations, the stationarity condition for the action $S(\psi)$ is given by 
\begin{equation*}
\begin{split}
\delta S(\psi)&=
\int_{U} \left\{ \delta p_{A}^{\mu} \left( \frac{\partial y^{A}}{\partial x^{\mu}}- v^{A}_{\mu} \right)+ 
 \left( -p_{A}^{\mu}+ \frac{\partial L}{\partial v^{A}_{\mu}} \right) \delta v^{A}_{\mu}+ 
p_{A}^{\mu}  \delta \left( \frac{\partial y^{A}}{\partial x^{\mu}}\right)+
 \frac{\partial L}{\partial y^{A}} \delta y^{A} \right\} d^{n+1}x\\
 &=\int_{U} \left\{ \delta p_{A}^{\mu} \left( \frac{\partial y^{A}}{\partial x^{\mu}}- v^{A}_{\mu} \right)+ 
 \left( -p_{A}^{\mu}+ \frac{\partial L}{\partial v^{A}_{\mu}} \right) \delta v^{A}_{\mu}+ 
 \left( -\frac{\partial p_{A}^{\mu}}{\partial x^{\mu}}+ \frac{\partial L}{\partial y^{A}}\right) \delta y^{A}
 \right\} d^{n+1}x\\
 &  + \int_{U} \frac{\partial }{\partial x^{\mu}} \left( p_{A}^{\mu} \delta y^{A}\right)d^{n+1}x\\
  &=\int_{U} \left\{ \delta p_{A}^{\mu} \left( \frac{\partial y^{A}}{\partial x^{\mu}}- v^{A}_{\mu} \right)+ 
 \left( -p_{A}^{\mu}+ \frac{\partial L}{\partial v^{A}_{\mu}} \right) \delta v^{A}_{\mu}+ 
 \left( -\frac{\partial p_{A}^{\mu}}{\partial x^{\mu}}+ \frac{\partial L}{\partial y^{A}}\right) \delta y^{A}
 \right\} d^{n+1}x\\
 &  + \int_{\partial U} \left( p_{A}^{\mu} \delta y^{A}\right)d^{n}x \\
 &=0.
\end{split}
\end{equation*}
By definition, the variation of $y^{A}$ vanishes at the boundary of $U$, namely, $\delta y^{A}\big|_{\partial U}=0$, and so it follows that
\begin{equation*}
\begin{split}
\int_{U} \left\{ \delta p_{A}^{\mu} \left( \frac{\partial y^{A}}{\partial x^{\mu}}- v^{A}_{\mu} \right)+ 
 \left( -p_{A}^{\mu}+ \frac{\partial L}{\partial v^{A}_{\mu}} \right) \delta v^{A}_{\mu}+ 
 \left( -\frac{\partial p_{A}^{\mu}}{\partial x^{\mu}}+ \frac{\partial L}{\partial y^{A}}\right) \delta y^{A}
 \right\} d^{n+1}x=0,
\end{split}
\end{equation*}
for all variations $\delta y^{A}$, $\delta v^{A}_{\mu}$ and $\delta p^{\mu}_{A}$. Thus we obtain equation \eqref{local-ELFieldEqn}.
\end{proof}

\paragraph{Generalized Energy Constraint.}
In addition to the implicit Euler-Lagrange equations given in equation \eqref{local-ELFieldEqn},
by imposing the {\it generalized energy constraint} 
\[
E=p+p_{A}^{\mu}v^{A}_{\mu}-L(x^{\mu},y^{A},v^{A}_{\mu})=0,
\]
we can naturally recover the covariant Legendre transformation:
\[
p=L(x^{\mu},y^{A},v^{A}_{\mu})-\frac{\partial L}{\partial v^{A}_{\mu}}v^{A}_{\mu}, \qquad 
p_{A}^{\mu}= \frac{\partial L}{\partial v^{A}_{\mu}}.
\]

\begin{theorem}
The following statements for a section $\psi : X \rightarrow M$ of $\pi_{XM}$ are equivalent:

\begin{itemize}
\item[(1)]  $\psi$ is a critical point of the Hamilton-Pontryagin action functional \eqref{HPActionFunc}.
\item[(2)]  $\psi^{\ast} (\mathcal{X}_{M} \! \text{\Large$\lrcorner$} \;\Omega_{\mathcal{E}})=0$ for all $\pi_{XM}$-vertical vector fields $\mathcal{V}_{M}$ on $M$.
\item[(3)] $\psi$ satisfies the implicit Euler-Lagrange equations \eqref{local-ELFieldEqn} together with the covariant Legendre transformation \eqref{CovLegTrans}.
\end{itemize}
\end{theorem}

\section{Multi-Dirac Structures}

\label{sec:multidirac}

In this section, we introduce the concept of {\bfi multi-Dirac structures} as a natural extension to the case of field theories of the concept of Dirac structures on manifolds developed by \cite{Cour1990}.   As a motivating example, we show that the graph of an (almost-)multisymplectic form on a manifold defines an almost multi-Dirac structure, and we relate the closedness of this form to the integrability of the corresponding multi-Dirac structure.  Later on, we will return to multi-Dirac structures with nonholonomic constraint distributions and show that they are integrable whenever the underlying constraint distribution is integrable.

Throughout this section, $M$ will be an arbitrary manifold. We will also fix an arbitrary integer $n + 1$ which we will refer to as the {\bfi degree} of the multi-Dirac structure, to be defined below.  In the remainder of the paper, we will consider the special case where $M$ is the Pontryagin bundle $M = J^1 Y \times_Y Z$.  In this case, $n + 1$ will be the dimension of $X$.

\paragraph{Pairings between Multivectors and Forms.}

Let $M$ be a manifold and consider the spaces $T^{l}M$ of $l$-multivector fields on $M$ and $\Lambda^{k}M$ of $k$-forms on $M$. For $k \ge l$, there is a natural pairing between elements of $\Sigma \in \Lambda^{k}M$ and $\mathcal{X} \in T^{l}M$ given by
\begin{equation}\label{pairing-mvf}
\left< \Sigma, \mathcal{X} \right> :=\mathbf{i}_{\mathcal{X}} \Sigma \in \Lambda^{k-l}(M).
\end{equation}

We now introduce the {\bfi graded Pontryagin bundle of degree $r$} over $M$ as follows:
\begin{equation} \label{prr}
P_{r} = T^{r}M \oplus \Lambda^{n+2-r}M,
\end{equation}
where $r=1,...,n+1$.   Note that $P_r$ also depends on the choice of $n$.  

Using the pairing in equation \eqref{pairing-mvf}, let us define the following  antisymmetric and symmetric pairings between the elements of $P_{r}$ and $P_{s}$ as follows.  For $(\mathcal{X}, \Sigma) \in P_{r}$ and $(\bar{\mathcal{X}}, \bar{\Sigma}) \in P_{s}$, where $r,s=1, \ldots ,n+1$, we put
\begin{equation} \label{gsym}
\left<\left< (\mathcal{X},  \Sigma), (\bar{\mathcal{X}}, \bar{\Sigma}) \right>\right>_{-} :=\frac{1}{2}\left( \mathbf{i}_{\bar{\mathcal{X}}}\Sigma-(-1)^{rs}\mathbf{i}_{\mathcal{X}}\bar{\Sigma} \right)
\end{equation}
and
\begin{equation} \label{gas}
\left<\left< (\mathcal{X},  \Sigma), (\bar{\mathcal{X}}, \bar{\Sigma}) \right>\right>_{+} :=\frac{1}{2}\left( \mathbf{i}_{\bar{\mathcal{X}}}\Sigma+(-1)^{rs}\mathbf{i}_{\mathcal{X}}\bar{\Sigma} \right),
\end{equation}
each of which takes values in $\Lambda^{n+2-r-s}(M)$.  Hence, both of these pairing are identically zero whenever $n+2<r+s$.

Let $V_{s}$ be a subbundle of $P_{s}$. The {\bfi $r$-orthogonal complementary subbundle} of $V_{s}$ associated to the antisymmetric paring $\left<\left<,\right>\right>_{-}$  is the subbundle $(V_{s})^{\perp,r}$ of $P_{r}$ defined by
\[
\begin{split}
(V_{s})^{\perp,r}&= \{  (\mathcal{X},  \Sigma) \in P_{r} \mid 
\left<\left< (\mathcal{X},  \Sigma), (\bar{\mathcal{X}}, \bar{\Sigma}) \right>\right>_{-}=0 \quad \textrm{for all} \quad (\bar{\mathcal{X}}, \bar{\Sigma}) \in V_{s}
\}.
\end{split}
\]
Note that $(V_{s})^{\perp,r}$ is a subbundle of $P_r$, and that  $(V_{s})^{\perp,r}=P_r$ whenever $n+2 < r+s$.

\paragraph{Almost Multi-Dirac Structures on Manifolds.}  

The definition of an \emph{almost multi-Dirac structure} on $M$ mimics the standard definition of \cite{Cour1990} of Dirac structures.  

\begin{definition}\label{DefMultiDirac} 
An {\bfi almost multi-Dirac structure of degree $n + 1$} on $M$ is a sequence of subbundles $D_1, \ldots, D_{n+1}$, where 
\[
D_{r} \subset  P_{r} \quad \text{for $r = 1, \ldots, n+1$},
\]
which is {\bfi $(n+1)$-Lagrangian}; namely, the sequence of the bundles $D_{r}$ satisfies the {\bfi maximally $(n+1)$-isotropic property}
\begin{equation} \label{isotropy}
D_{r}=(D_{s})^{\perp, r}
\end{equation}
for $r, s=1,...,n+1$ and where $r+s \le n+2$. 
\end{definition}

When no confusion can arise, we will refer to the sequence $D_1, \ldots, D_{n+1}$ simply as $D$.  Later, we will define an \emph{integrable multi-Dirac structure} as an almost multi-Dirac structure that satisfies certain integrability conditions.  

For the case of classical field theories, it will turn out that only the orthogonal complements of the multi-Dirac structure of the ``lowest'' and ``highest'' order (namely, $r=1$ and $r=n+1$ respectively) play an essential role to formulate the field equations.

\paragraph{Example.}
Let $M$ be a manifold with a multi-Dirac structure $D_1$ of degree $1$.  The isotropy property then becomes
\[
	D_1 = (D_1)^{\perp, 1},  
\]
where $P_1 = TM \oplus T^\ast M$ and the pairing is given by 
\[
	\left<\left< (\mathcal{X},  \Sigma), (\bar{\mathcal{X}}, \bar{\Sigma}) \right>\right>_{-} =
	\frac{1}{2}\left( \mathbf{i}_{\bar{\mathcal{X}}}\Sigma + \mathbf{i}_{\mathcal{X}}\bar{\Sigma} \right).
\]
for all $(\mathcal{X}, \Sigma), (\bar{\mathcal{X}}, \bar{\Sigma}) \in P_1$.  This is nothing but the definition of an (almost) Dirac structure developed by \cite{Cour1990}.  Our concept of multi-Dirac structures also includes the so-called \emph{higher-order Dirac structures} of \cite{Zambon2010}.    We will return to this at the end of this section.

\paragraph{Multi-Dirac Structures Induced by Differential Forms.}

Consider an $(n+2)$-form $\Omega_{M}$ on a manifold $M$.  We will show that the graph of $\Omega_M$ (in the sense defined below) defines a multi-Dirac structure of degree $n + 1$.  This example of a multi-Dirac structure will be fundamental in our subsequent treatment of classical field theories, where $\Omega_M$ will be the canonical multisymplectic form, but for now $\Omega_M$ can be an arbitrary form.  Note especially that at this stage $\Omega_M$ does not need to be closed or non-degenerate.

\begin{proposition} \label{prop: canondirac}
Let $\Omega_{M}$ be an arbitrary $(n + 2)$-form on $M$ and define the following subbundles $D_1, \ldots, D_{n + 1}$, where $D_r \subset P_r$:
\begin{equation}
\label{canonms}
D_{r} =\left\{ (\mathcal{X}, \Sigma)  \in P_{r}  \mid  
 \;\mathbf{i}_{\mathcal{X}} \Omega_{M}=\Sigma  \right\}
\end{equation}
for $r = 1, \ldots, n + 1$.  Then $D = D_1, \ldots, D_{n + 1}$ is a  multi-Dirac structure of degree $n + 1$ on $M$.
\end{proposition}

\begin{proof}
To prove that $D = D_1, \ldots, D_{n + 1}$ is a multi-Dirac structure of degree $n + 1$, we need to check the isotropy property \eqref{isotropy}, namely namely, $D_{r}=D_{s}^{\perp,r}$ for all $r, s = 1, \ldots, n + 1$, with $r + s \le n + 1$.

Let us first show that $D_{r} \subset D_{s}^{\perp,r}$. Let $(\mathcal{X},\Sigma) \in D_{r}$
and $(\bar{\mathcal{X}},\bar{\Sigma}) \in D_{s}$. Then, it follows that
\[
\begin{split}
&\left<\left< (\mathcal{X},\Sigma),(\bar{\mathcal{X}},\bar{\Sigma})\right>\right>_{-}\\
&\hspace{1cm}=\frac{1}{2}\left\{\mathbf{i}_{\bar{\mathcal{X}}} \Sigma +(-1)^{rs+1}\mathbf{i}_{\mathcal{X}} \bar{\Sigma}\right\}\\
&\hspace{1cm}=\frac{1}{2}\left\{\mathbf{i}_{\bar{\mathcal{X}}}\mathbf{i}_{\mathcal{X}} \Omega_{M} +(-1)^{rs+1} \mathbf{i}_{\mathcal{X}}\mathbf{i}_{\bar{\mathcal{X}}}\Omega_{M}\right\}
\\
&\hspace{1cm}=0,
\end{split}
\]
since $\mathbf{i}_{\bar{\mathcal{X}}}\mathbf{i}_{\mathcal{X}}\Omega_{M}=(-1)^{rs} \mathbf{i}_{\mathcal{X}}\mathbf{i}_{\bar{\mathcal{X}}}\Omega_{M}$.
Thus, $D_{r} \subset D_{s}^{\perp,r}$.
\medskip

Next, let us show that $D_{s}^{\perp,r} \subset D_{r}$. Let $(\bar{\mathcal{X}}, \bar{\Sigma}) \in D_{s}^{\perp,r}$. By definition of $D_{s}^{\perp,r}$,
\[
\mathbf{i}_{\bar{\mathcal{X}}} \Sigma + (-1)^{rs+1}\mathbf{i}_{\mathcal{X}} \bar{\Sigma}=0
\] 
for all $(\mathcal{X},\Sigma)  \in D_{r}$, i.e. 
$\mathcal{X} \in T^{r}Z$ such that $\mathbf{i}_{\mathcal{X}} \Omega_{M}=\Sigma$.  Then, it follows
\[
\begin{split}
\mathbf{i}_{\bar{\mathcal{X}}} \Sigma + (-1)^{rs+1}\mathbf{i}_{\mathcal{X}} \bar{\Sigma} &=\mathbf{i}_{\bar{\mathcal{X}}} \mathbf{i}_{\mathcal{X}} \Omega_{M} + (-1)^{rs+1}\mathbf{i}_{\mathcal{X}} \bar{\Sigma} \\
&=\mathbf{i}_{\mathcal{X}} \left\{  (-1)^{rs} \mathbf{i}_{\bar{\mathcal{X}}} \Omega_{M}+ (-1)^{rs+1} \bar{\Sigma}\right\}\\
&=0
\end{split}
\]
for all $\mathcal{X} \in T^{s}Z$ and with $r+s \le n+2$. Therefore, one has
\[
\mathbf{i}_{{\bar{\mathcal{X}}}}\Omega_{M}=\bar{\Sigma}.
\]

Thus, $D_{s}^{\perp, r} \subset D_{r}$. Finally, we have shown that $D_{r}=D_{s}^{\perp, r}$. So, it follows that $D_{r}$ is a multi-Dirac structure of degree $r$ on $M$.
\end{proof}

\begin{proposition}
Let $D = D_1, \ldots, D_{n+1}$ be a multi-Dirac structure of degree $n + 1$.  For any $(\mathcal{X},\Sigma) \in D_{r}$ and $(\bar{\mathcal{X}},\bar{\Sigma}) \in D_{s}$, the following relation holds:
\begin{equation}\label{LagProp}
 \mathbf{i}_{\bar{\mathcal{X}}}\Sigma-(-1)^{rs}\mathbf{i}_{\mathcal{X}}\bar{\Sigma} =0.
\end{equation}
\end{proposition}
\begin{proof}
This is clear from the $r$-Lagrangian (maximally $r$-isotropic) property of $D_{r}$. 
\end{proof}

The above $r$-Lagrangian property of the multi-Dirac structure $D_{r}$ in equation \eqref{LagProp} may be understood as the field-theoretic analogue of the {\bfi virtual work principle} in mechanics and is related to {\bfi Tellegen's theorem} in electric circuits. We shall return to this relation in the construction of Lagrange-Dirac field theories.

\paragraph{Wedge Products on Sections of Graded Pontryagin Bundles.}
We can introduce a wedge product between elements of the space of sections $P_{r}$ and $P_{s}$, which is given by, for $(\mathcal{X},\Sigma) \in P_{r}$
and $(\bar{\mathcal{X}},\bar{\Sigma}) \in P_{s}$,
\begin{equation} \label{wedgeproduct}
\begin{split}
(\mathcal{X},\Sigma) \wedge (\bar{\mathcal{X}},\bar{\Sigma})
&:=
\left( \mathcal{X} \wedge \bar{\mathcal{X}},\left<\left< (\mathcal{X},  \Sigma), (\bar{\mathcal{X}}, \bar{\Sigma}) \right>\right>_{+} \right)\\
&=\left( \mathcal{X} \wedge \bar{\mathcal{X}},\frac{1}{2}\left( \mathbf{i}_{\bar{\mathcal{X}}}\Sigma+(-1)^{rs}\mathbf{i}_{\mathcal{X}}\bar{\Sigma} \right) \right) \in P_{r+s},.
\end{split}
\end{equation}
where on the right-hand side we have used the usual wedge product of multi-vector fields.

\paragraph{Multi-Courant Brackets of Forms and Multi-Vector Fields.}
We now introduce a notion of multi-Courant brackets which is a natural extension to the case of field theories of the bracket used by \cite{Cour1990}.  To do this,
we define first the {\bfi Lie derivative} $\pounds_{\mathcal{X}} \Sigma$ of an $l$-form $\Sigma$ with respect to a $k$-multivector field $\mathcal{X}$ by means of Cartan's magic formula as
\[
\pounds_{\mathcal{X}} \Sigma := 
\mathbf{i}_{\mathcal{X}} \mathbf{d}\Sigma - (-1)^k
\mathbf{d} \mathbf{i}_{\mathcal{X}} \Sigma.
\]
Note that $\pounds_{\mathcal{X}} \Sigma$ is an $(l+1-k)$-form (see \cite{Tu1974} for more information). 
\medskip

Using this notion of Lie derivative, we define a {\bfi multi-Courant bracket} on the space of sections of $P_{r} \times P_{s}$ as a graded anti-symmetric bracket
\[
\left[\!\left[ \cdot, \cdot \right]\!\right]_{r,s}: \Gamma(P_{r}) \times \Gamma(P_{s}) \to \Gamma(P_{r+s-1})
\]
given by
\begin{equation} \label{multicourant}
\begin{split}
&\left[\!\!\left[\left(\mathcal{X},\Sigma\right), \left(\bar{\mathcal{X}}, \bar{\Sigma}\right) \right]\!\!\right]_{r,s}\\
&:=
\left( [\mathcal{X},\bar{\mathcal{X}}], \; \pounds_{\mathcal{X}}\bar{\Sigma}-(-1)^{(r-1)(s-1)}\pounds_{\bar{\mathcal{X}}}\Sigma+\frac{(-1)}{2}^{r}\mathbf{d}
\left<\!\left<(\mathcal{X}, \Sigma),(\bar{\mathcal{X}}, \bar{\Sigma})\right>\!\right>_{+}
\right)\\
& \phantom{:} =
\left( [\mathcal{X},\bar{\mathcal{X}}], \; \pounds_{\mathcal{X}}\bar{\Sigma}-(-1)^{(r-1)(s-1)}\pounds_{\bar{\mathcal{X}}}\Sigma+\frac{(-1)}{2}^{r}\mathbf{d}
\left( \mathbf{i}_{\bar{\mathcal{X}}}\Sigma+(-1)^{rs}\mathbf{i}_{\mathcal{X}}\bar{\Sigma} \right)
\right).
\end{split}
\end{equation}
In the above, $[\mathcal{X},\bar{\mathcal{X}}]$ is the {\it Schouten-Nijenhuis bracket} of the $r$-multivector  field $\mathcal{X}$ and the $s$-multivector field $\bar{\mathcal{X}}$, so that $[\mathcal{X},\bar{\mathcal{X}}]$ is an $(r+s-1)$-multivector field. 

 For the case in which $n=0$ and $r=s=1$, the multi-Courant bracket corresponds to the standard Courant bracket  on the space $\Gamma(P)$ of sections of the Pontryagin bundle $P=TM\oplus T^{\ast}M$.

\begin{definition} \label{def:IntMultDirac}
An almost multi-Dirac structure $D = D_1, \ldots, D_{n + 1}$  of degree $n+1$ on $M$ is said to be {\bfi integrable} if for all $(\mathcal{X},\Sigma) \in D_{r}$ and $(\bar{\mathcal{X}},\bar{\Sigma}) \in D_{s}$, where $r, s=1,...,n+1$ and $r+s \le n+1$, the \newtext{following condition is satisfied:
\[
\left[\!\!\left[\left(\mathcal{X},\Sigma\right), \left(\bar{\mathcal{X}}, \bar{\Sigma}\right) \right]\!\!\right]_{r,s} \in D_{r+s-1}
\]}
\end{definition}

\paragraph{Integrable Multi-Dirac Structures.}
In the construction of the almost multi-Dirac structure in Proposition \ref{prop: canondirac}, we only used the fact that the pre-multisymplectic form $\Omega_{M}$ maps multivectors $\mathcal{X}$ into forms $\Sigma = \mathbf{i}_{\mathcal{X}} \Omega_{M}$.  The fact that $\mathbf{d} \Omega_{M} = 0$ was left out of consideration, and in fact \emph{any} $(n + 2)$-form can be used to define a multi-Dirac structure.  We will now show that the multi-Dirac structure defined by an $(n + 2)$-form $\Omega$ is integrable if and only if $\Omega$ is closed.

Let $D = D_1, \ldots, D_{n + 1}$ be the multi-Dirac structure defined by an $(n + 2)$-form $\Omega$ as in \eqref{canonms}.  \newtext{For $D$ to be integrable, we need to check that the multi-Courant bracket \eqref{multicourant} is closed.  Note that the wedge product is always closed:} for all  $(\mathcal{X},\Sigma) \in D_{r}$ and $(\bar{\mathcal{X}},\bar{\Sigma}) \in D_{s}$, we have
\[
(\mathcal{X},\Sigma) \wedge (\bar{\mathcal{X}},\bar{\Sigma})=
(\mathcal{X}, \mathbf{i}_{\mathcal{X}}\Omega_{M}) \wedge (\bar{\mathcal{X}}, \mathbf{i}_{\bar{\mathcal{X}}}\Omega_{M}) 
=(\mathcal{X} \wedge \bar{\mathcal{X}}, \mathbf{i}_{\bar{\mathcal{X}}}\mathbf{i}_{\mathcal{X}}\Omega_{M}) 
=(\mathcal{X} \wedge \bar{\mathcal{X}}, \mathbf{i}_{\mathcal{X} \wedge \bar{\mathcal{X}}}\Omega_{M}),
\]
which is an element of $D_{r + s}$.  We now need to check the closedness of the multi-Courant bracket \eqref{multicourant}; namely,
\[
\left[\!\!\left[\left(\mathcal{X},\Sigma\right), \left(\bar{\mathcal{X}}, \bar{\Sigma}\right) \right]\!\!\right]_{r,s} \in D_{r+s-1}.
\]
This condition turns out to be equivalent to the closedness of $\Omega$, as is shown in the next theorem.

\begin{theorem} \label{thm:integrable}
Let $\Omega_{M}$ be an arbitrary $(n + 2)$-form on $M$ which is not necessarily closed and let $D$ be the almost multi-Dirac structure defined by the graph of $\Omega_{M}$.  Then, we have that $D$ is integrable if and only if $\mathbf{d} \Omega_{M} = 0$.
\end{theorem}
\begin{proof}
For the almost multi-Dirac structure induced by an $(n + 2)$-form $\Omega_{M}$, the multi-Courant bracket can be rewritten as follows.  Let $(\mathcal{X},\Sigma)$ and $(\bar{\mathcal{X}}, \bar{\Sigma})$ be sections of $D_{r}$ and $D_{s}$. Using the fact that, for the almost multi-Dirac structure, 
\[
	\Sigma = \mathbf{i}_{\mathcal{X}} \Omega_{M}	\quad \textrm{and} \quad \bar{\Sigma} = \mathbf{i}_{\bar{\mathcal{X}}} \Omega_{M},
\]	
we obtain after some calculations that 
\begin{align*}
[\![ (\mathcal{X}, \Sigma),(\bar{\mathcal{X}}, \bar{\Sigma})]\!]_{r,s} & =  ([\mathcal{X}, \bar{\mathcal{X}}], \pounds_{\mathcal{X}}\bar{\Sigma}  - (-1)^{(r-1)(s-1)}\pounds_{\bar{\mathcal{X}}} \Sigma +
 \frac{(-1)^{r}}{2} \mathbf{d} (\mathbf{i}_{\mathcal{X}} \bar{\Sigma} +(-1)^{rs}\mathbf{i}_{\bar{\mathcal{X}}}\Sigma)
) \\
& = ([\mathcal{X}, \bar{\mathcal{X}}], \pounds_{\mathcal{X}} \mathbf{i}_{\bar{\mathcal{X}}} \Omega_{M} - (-1)^{(r-1)(s-1)} \mathbf{i}_{\bar{\mathcal{X}}} \pounds_{\mathcal{X}} \Omega_{M} +  (-1)^{r} \mathbf{i}_{\mathcal{X}} \mathbf{i}_{\bar{\mathcal{X}}} \mathbf{d}\Omega_{M} ) \\
& = ([\mathcal{X}, \bar{\mathcal{X}}],  \mathbf{i}_{[\mathcal{X}, \bar{\mathcal{X}}]}\Omega_{M} +  (-1)^{r} \mathbf{i}_{\mathcal{X}} \mathbf{i}_{\bar{\mathcal{X}}}  \mathbf{d}\Omega_{M} ).
\end{align*}
In order for the right-hand side of this expression to be a section of $D_{r+s}$, the second term has to vanish.  This is precisely equivalent to the requirement that $\mathbf{d}\Omega_{M} = 0$. 
\end{proof}

\paragraph{Multi-Poisson Bracket Induced by a Multi-Dirac Structure.}

Associated to every Dirac structure there is a Poisson structure on a restricted class of functions, which satisfies the Jacobi identity if and only if the Dirac structure is integrable.  This correspondence was investigated by \cite{Cour1990} and forms the basis for the constrained Dirac brackets of \cite{VaMa1994} and \cite{CaLeMa1999}.  We now show that multi-Dirac structures similarly induce a notion of {\bfi multi-Poisson structure}, which is a graded analogue of the usual concept of Poisson structures.  Multi-Poisson brackets are studied almost exclusively in the context of multisymplectic structures (see among others \cite{CIdeLe96, CLM2003, FoPaRo2005} and the references therein, or \cite{BaHoRo2010, Rogers2010a} for a category-theoretic approach) and no comprehensive general theory exists as of yet.

Let $D = D_1, \ldots, D_{n + 1}$ be a multi-Dirac structure of degree $n + 1$.  For $k = 0, \ldots, n$, we say that a $k$-form $\Sigma$ is {\bfi admissible} if there exists an $(n + 1 - k)$-multivector field $\mathcal{X}_\Sigma$ such that
\[
	(\mathcal{X}_\Sigma, \mathbf{d} \Sigma) \in D_{n + 1 - k}.
\]
We denote the space of admissible $k$-forms by $\Omega^{k}_{\mathrm{adm}}(M)$.  We now define the {\bfi multi-Poisson structure} on $\Omega^\ast_{\mathrm{adm}}$ associated to $D$ as the map 
\[
	\{ \cdot, \cdot \} :  \Omega^k_{\mathrm{adm}} \times \Omega^l_{\mathrm{adm}} \rightarrow \Omega^{k + l  - n}_{\mathrm{adm}},
\]
(where $k, l = 0, \ldots, n$) given by the following prescription.  For $\Sigma \in \Omega^k_{\mathrm{adm}}$ and $\bar{\Sigma} \in \Omega^l_{\mathrm{adm}}$, we put 
\begin{equation} \label{semibracket}
	\{ \Sigma, \bar{\Sigma} \} := 
		\mathbf{i}_{\mathcal{X}_{\bar{\Sigma}}} \mathbf{d} \Sigma.
\end{equation}
where $\mathcal{X}_{\bar{\Sigma}}$ is chosen so that $(\mathcal{X}_{\bar{\Sigma}}, \mathbf{d} \bar{\Sigma}) \in D_{n + 1- l}$.  It can easily be shown that the bracket does not depend on the choice of multivector $\mathcal{X}_{\bar{\Sigma}}$: let $\mathcal{X}'_{\bar{\Sigma}}$ be any other multivector such that $(\mathcal{X}'_{\bar{\Sigma}}, \mathbf{d} \bar{\Sigma}) \in D_{n +1 -l}$.  The difference $(\mathcal{X}_{\bar{\Sigma}} - \mathcal{X}'_{\bar{\Sigma}}, 0)$ is then also an element of $D_{n - l}$ and therefore we have that 
\[
	\left<\left< (\mathcal{X}, \Sigma), (\mathcal{X}_{\bar{\Sigma}} - \mathcal{X}'_{\bar{\Sigma}}, 0) \right>\right>_{-} = 0 
\]
but this is equivalent to 
\[
\mathbf{i}_{\mathcal{X}_{\bar{\Sigma}}} \mathbf{d} \Sigma
=
\mathbf{i}_{\mathcal{X}'_{\bar{\Sigma}}} \mathbf{d} \Sigma
\]
which shows us that the bracket $\{ \Sigma, \bar{\Sigma} \}$ depends only on the choice of $\Sigma$ and $\bar{\Sigma}$.  \newtext{ At this stage, we remark that the designation ``Poisson'' is chosen to emphasize the similarity with standard Poisson structures.  However, to obtain a (graded) Poisson algebra we would also have to introduce a product on $\Omega^\ast_{\mathrm{adm}}$ with respect to which the Poisson bracket is a (graded) derivation.   A candidate graded product is given by the wedge product \eqref{wedgeproduct}, but a detailed study of these structures is deferred to a forthcoming paper.} 

To finish this paragraph, we remark that for a multi-Dirac structure induced by a multisymplectic form $\Omega$, the bracket $\{ \Sigma, \bar{\Sigma} \}$ agrees with the multi-Poisson brackets of \cite{CIdeLe96, FoPaRo2005, BaHoRo2010}.  Indeed, let $D = D_1, \ldots, D_{n + 1}$ be a multi-Dirac structure induced by a multisymplectic form.  A $k$-form $\Sigma$ is admissible if there exists an $(n + 1 - l)$-multivector field $\mathcal{X}_\Sigma$ so that $(\mathcal{X}_\Sigma, \Sigma) \in D_{n + 1 - l}$.  This is equivalent to 
\[
	\mathbf{i}_{\mathcal{X}_\Sigma} \Omega_M = \mathbf{d} \Sigma.
\]
Forms with this property are said to be {\bfi Hamiltonian}.  On this class of forms, the bracket \eqref{semibracket} is given by 
\begin{equation} \label{semims}
	\{ \Sigma, \bar{\Sigma} \} = \mathbf{i}_{\mathcal{X}_{\bar{\Sigma}}}
		\mathbf{i}_{\mathcal{X}_{\Sigma}} \Omega_M,
\end{equation}
which agrees (up to sign) with the definition of the brackets in the multisymplectic literature.

It is worth noting that one often encounters two different, non-equivalent multi-Poisson brackets in classical field theory.  The first one, referred to as the {\bfi semi-bracket} in \cite{BaHoRo2010} is defined as in \eqref{semims}, while the {\bfi hemi-bracket} is defined for two Hamiltonian forms as $\{ \Sigma, \bar{\Sigma} \}' :=  \pounds_{\mathcal{X}_{\bar{\Sigma}}} \Sigma$.  However, this notion of bracket does not seem to be well-defined for arbitrary multi-Dirac structures, since the right-hand side depends on the choice of multivector field $\mathcal{X}_{\bar{\Sigma}}$.  In order words, only the semi-bracket can be defined for multi-Dirac structures.  As shown by \cite{Zambon2010}, a similar phenomenon appears in the theory of higher-order Dirac structures, to be discussed below.

Finally, \cite{BaHoRo2010} have shown that on a multisymplectic manifold $(M, \Omega)$ where the degree of $\Omega$ is three, the space $\Omega^2_{\mathrm{adm}}$ has the structure of a {\bfi Lie 2-algebra}.  This observation was further extended to the case of higher-order Dirac structures by \cite{Zambon2010} and to the case of multisymplectic structures of arbitrary degree by \cite{Rogers2010b}.  Following a conjecture of \cite{Zambon2010}, it follows that our spaces $\Omega^n_{\mathrm{adm}}(M)$ of highest-degree admissible forms are Lie $(n + 1)$-algebras, but more is true.  When considering the space of admissible forms $\Omega^\ast_{\mathrm{adm}}(M)$ in its entirety, we can define a graded \newtext{Lie}  bracket on it as in \cite{CIdeLe96} as follows. For $\alpha \in \Omega^k_{\mathrm{adm}}(M)$, we put 
\[
	\left| \alpha \right| = n - k, 
\]
and we define $\tilde{\Omega}^k_{\mathrm{adm}}(M) := \Omega^{n - k}_{\mathrm{adm}}(M)$ for $k = 0, \ldots, n - 1$.  The multi-Poisson bracket is then an operator 
\[
	\{\cdot, \cdot \} : 
	\tilde{\Omega}^k_{\mathrm{adm}}(M) \times 
	\tilde{\Omega}^l_{\mathrm{adm}}(M) \rightarrow 
	\tilde{\Omega}^{k + l}_{\mathrm{adm}}(M).
\]
For multisymplectic manifolds, \cite{CIdeLe96} show that $\tilde{\Omega}^\ast_{\mathrm{adm}}$ is endowed with the structure of a graded \newtext{Lie algebra}, where the graded Jacobi identity is satisfied up to closed forms.  Their proof can be easily adapted to the case of multi-Dirac structures, but we defer a comprehensive discussion of these results to a forthcoming paper.

\paragraph{Relation with Higher-Order Dirac Structures.}   

We now show that multi-Dirac structures are naturally related to \emph{higher-order Dirac structures} as defined by \cite{Zambon2010}.  A higher-order Dirac structure of order $p$ on a manifold $M$ is an isotropic subbundle of $TM \oplus \Lambda^p(M)$ equipped with the following pairing: 
\begin{equation} \label{hopairing}
	\left<\!\left< (X, \alpha), (Y, \beta) \right>\!\right> 
		= \mathbf{i}_X \beta + \mathbf{i}_Y \alpha.
\end{equation}
and which is involutive under the \emph{Courant-Dorfmann bracket}:
\begin{equation} \label{CDbracket}
	\left[\!\left[ (X, \alpha), (Y, \beta) \right]\!\right] =
		( [X, Y], \pounds_X \beta - \mathbf{i}_Y \alpha). 
\end{equation}

Let $D = D_1, \ldots, D_{n + 1}$ be an integrable multi-Dirac structure of degree $n + 1$.  Recall that the multi-Courant bracket is a map from $D_r \times D_s$ to $D_{r + s - 1}$, so that $D_1$ is closed under the multi-Courant bracket in its own right.  We claim that $D_1$ is a higher-order Dirac structure of degree $n + 1$.  First of all, note that $D_1 \subset P_1 = TM \oplus \Lambda^{n+1}(M)$.  The antisymmetric pairing \eqref{gsym}, restricted to $P_1$, is given by 
\[
\left<\left< (\mathcal{X},  \Sigma), (\bar{\mathcal{X}}, \bar{\Sigma}) \right>\right>_{-}=\frac{1}{2}\left( \mathbf{i}_{\bar{\mathcal{X}}}\Sigma + \mathbf{i}_{\mathcal{X}}\bar{\Sigma} \right),
\]
which agrees with \eqref{hopairing} up to an insignificant multiplicative factor of $1/2$.  Secondly, the multi-Courant bracket restricted to sections of $D_1$ is given by 
\[
\left[\!\!\left[\left(\mathcal{X},\Sigma\right), \left(\bar{\mathcal{X}}, \bar{\Sigma}\right) \right]\!\!\right]_{1,1}\\
=
\left( [\mathcal{X},\bar{\mathcal{X}}], \; \pounds_{\mathcal{X}}\bar{\Sigma}-\pounds_{\bar{\mathcal{X}}}\Sigma+
\mathbf{d} \mathbf{i}_{\bar{\mathcal{X}}} \Sigma \right),
\]
which is nothing but the Courant-Dorfmann bracket \eqref{CDbracket}.

\newtext{
In this way, we proved that the correspondence between multi-Dirac and higher-order Dirac structures, given by mapping $D_1, \ldots, D_{n + 1}$ to $D_1$ is injective.  It is easy to show that this mapping preserves the Courant-Dorfmann bracket, so that integrable multi-Dirac structures are mapped to integrable higher-order Dirac structures.  \cite{Zambon2010}[Prop.~4.2] shows that this correspondence is a bijection, so that multi-Dirac and higher-order Dirac structures are \emph{equivalent}.
}

\section{Lagrange-Dirac Field Theories}

We now consider Lagrangian field theories in the context of multi-Dirac structures.
Recall that $X$ is an oriented manifold with $\dim X = n + 1$ and with a fixed volume form $\eta$, locally given by $\eta=d^{n+1}x$. Recall also that $\pi_{XY}: Y \rightarrow X$ is a fiber bundle and consider a Lagrangian density $\mathcal{L}$ on $J^1 Y$.  Let $\Omega$ be the canonical multisymplectic structure on $Z$ and recall the pre-multisymplectic $(n+2)$-form defined on the Pontryagin bundle  $M=J^{1}Y \oplus Z$ by $\Omega_{M}=\pi_{ZM}^{\ast}\Omega$.
  
Using $\Omega_M$ we can define a multi-Dirac structure $D = D_1, \ldots, D_{n + 1}$ of degree $n + 1$ on $M$ by the construction of proposition~\ref{prop: canondirac}.  We refer to $D$ as the {\bfi canonical multi-Dirac structure} on $M$.  Explicitly, $D = D_1, \ldots, D_{n + 1}$ is given by 
\begin{equation} \label{canondirac}
	D_r = \{ (\mathcal{X}, \mathbf{i}_{\mathcal{X}} \Omega_M): \quad \mathcal{X} \in T^r M \}.
\end{equation}
for $r = 1, \ldots, n + 1$.  

In this section, we introduce the concept of {\bfi Lagrange-Dirac field theories}, which are field theories whose field equations are specified in terms of a multi-Dirac structure.   More precisely, we shall see that only the component $D_{n + 1} \subset T^{n + 1} M \times \Lambda^1(M)$ of the multi-Dirac structure is needed for the formulation of the field equations.   We then show that the implicit Euler-Lagrange equations \eqref{ELFieldEqn-M} for an example of a Lagrange-Dirac field theory, as they can be described in terms of the canonical multi-Dirac structure $D$.  In the next section, we will address the case of field theories with {\it nonholonomic constraints}.  These field theories can be interpreted as Lagrange-Dirac systems as well, where now the nonholonomic constraints are incorporated in the specification of the multi-Dirac structure.

\paragraph{Partial Multi-vector Fields.}

Let $\mathcal{X}$ be an $(n + 1)$-multivector field on $M$.  We say that $\mathcal{X}$ is a {\bfi partial multivector field} if $\mathcal{X}$ satisfies 
\[
	T^{n + 1} \pi_{ZM} \circ \mathcal{X} = 0,
\]
where $T^{n + 1} \pi_{ZM} : T^{n + 1} M \rightarrow T^{n + 1} Z$ is the natural projection.  In other words, $\mathcal{X}$ is partial if it does not have a component along the $Z$-direction.  In the remainder of this section, we will always consider partial multivector fields which are \emph{integrable}, \emph{decomposable} and satisfy the \emph{normalization condition} 
\[
	\mathbf{i}_{\mathcal{X}}(\pi_{XZ}^{\ast}\eta)=1.
\]
See section~\ref{sec:geometry} for a definition of these properties.  Locally, such multivector fields can be written as 
\begin{equation} \label{multvectf}
	\bar{\mathcal{X}} = \bigwedge_{\mu = 1}^{n + 1} \bar{\mathcal{X}}_{\mu} = \bigwedge_{\mu = 1}^{n + 1}
		\left( \frac{\partial}{\partial x^\mu}
			+ C^A_\mu \frac{\partial}{\partial y^A}
			+ C_{A\mu}^\nu \frac{\partial}{\partial p_A^\nu}
			+ C_\mu \frac{\partial}{\partial p} \right).
\end{equation}
where $C^A_\mu$, $C_{A\mu}^\nu$ and $C_\mu$ are the local component functions of $\mathcal{X}$ on $M$.  Note that the component of the multivector field along $\partial/\partial{v^{A}_{\mu}}$ is zero.

\paragraph{Lagrange-Dirac Field Theories.}  

We now define a special class of field theories, whose field equations are specified in terms of a Lagrangian $\mathcal{L}$ and a multi-Dirac structure $D$ on the Pontryagin bundle $M$.

\begin{definition}
Let $D = D_1, \ldots, D_{n + 1}$ be a multi-Dirac structure of degree $n + 1$ on the Pontryagin bundle $M$ and consider a Lagrangian density $\mathcal{L}$ with associated generalized energy $E$ given by \eqref{generalizedenergy}.  A {\bfi Lagrange-Dirac system for field theories} is defined by a triple $(\mathcal{X},E, D_{n+1})$, where $\mathcal{X}$ is a normalized, decomposable partial vector field, so that 
\begin{equation}\label{cond_ILS_M}
(\mathcal{X}, (-1)^{n+2}\mathbf{d}E) \in D_{n+1}.
\end{equation}
\end{definition}

Note that in the definition of a Lagrange-Dirac system, only the highest-order component $D_{n + 1}$ of the multi-Dirac structure appears.  

At this point, we stress that the multi-Dirac structure $D$ is arbitrary.  However, when $D$ is the canonical multi-Dirac structure induced by $\Omega_M$, the resulting Lagrange-Dirac equations \eqref{cond_ILS_M} are nothing but the implicit Euler-Lagrange equations \eqref{ELFieldEqn-M} obtained from the Hamilton-Pontryagin variational principle, as we now show.

\begin{theorem} \label{thm:LDeq} 
Let $D = D_1, \ldots, D_{n + 1}$ be the canonical multi-Dirac structure on $M$ induced by $\Omega_M$ as in \eqref{canondirac}.  The Lagrange-Dirac system $(\mathcal{X},E,D_{n+1})$ in equation  \eqref{cond_ILS_M} induces {\bfi implicit Euler-Lagrange equations for field theories}:
\begin{equation} \label{imp_EL}
\mathbf{i}_{\mathcal{X}}\Omega_{M}=(-1)^{n+2}\mathbf{d}E, \qquad \mathbf{i}_{\mathcal{X}}(\pi_{XM}^{\ast}\eta)=1,
\end{equation}
which are written in local coordinates as
\begin{equation} \label{eqm-LDS}
	\frac{\partial p_A^\mu}{\partial x^\mu} =  \frac{\partial L}{\partial y^A}, \quad
	\frac{\partial y^A}{\partial x^\mu} = v^A_\mu, \quad
	p_A^\mu = \frac{\partial L}{\partial v^A_\mu},
\end{equation}
together with
\begin{equation} \label{encons-LDS}
	\frac{\partial}{\partial x^\mu} \left( p+p_A^\nu v^A_\nu - L \right) = 0.
\end{equation}

\end{theorem}

\begin{proof}
Recall the pre-multisymplectic form $\Omega_{M}$ is locally given by
\[
\Omega_{M}=dy^{A} \wedge dp_{A}^{\mu}\wedge d^{n}x_{\mu}-dp \wedge d^{n+1}x.
\]
Note that the fiber coordinates $y^A, v^A_\mu, p_A^\mu, p$ appear at most twice in this expression, so that the contraction of $\Omega_M$ with three or more $\pi_{XM}$-vertical vector fields vanishes.

Now, the differential of the generalized energy 
\[
E(x^{\mu},y^{A},v^{A}_{\mu},p_{A}^{\mu},p)=p+p_{A}^{\mu}v^{A}_{\mu}-L(x^{\mu},y^{A},v^{A}_{\mu})
\]
on $M$ is given by
\[
\begin{split}
\mathbf{d}E&=\frac{\partial E}{\partial x^{\mu}}dx^{\mu}+\frac{\partial E}{\partial y^{A}}dy^{A}+\frac{\partial E}{\partial v^{A}_{\mu}}dv^{A}_{\mu}+\frac{\partial E}{\partial p_{A}^{\mu}}dp_{A}^{\mu}+\frac{\partial E}{\partial p}dp\\
&=\left(-\frac{\partial L}{\partial x^{\mu}}\right)dx^{\mu}+\left(-\frac{\partial L}{\partial y^{A}} \right)dy^{A}+\left(p_{A}^{\mu}-\frac{\partial L}{\partial v^{A}_{\mu}} \right)dv^{A}_{\mu}
+v^{A}_{\mu}dp_{A}^{\mu}+dp.
\end{split}
\]

The local expression for the $(n + 1)$-partial multivector field $\mathcal{X}$ is given by \eqref{multvectf}, and a computation in local coordinates (see section~\ref{sec:comp}) shows that 
\begin{equation} \label{contract}
	\mathbf{i}_\mathcal{X} \Omega_{M} = 
		(-1)^{n+2} \left[
			(C^A_\mu C_{A\lambda}^\lambda - 
			 C^A_\lambda C^\lambda_{A\mu} - C_\mu) dx^\mu 
			 + C^A_\mu dp_A^\mu + C_{A\mu}^\mu dy^A + dp
		\right].
\end{equation}
As a result, we have that \eqref{imp_EL} holds if and only if 
\begin{equation} \label{comp1}
	C^A_\mu C_{A\lambda}^\lambda - C^A_\lambda C^\lambda_{A\mu}
	- C_\mu =  -\frac{\partial L}{\partial x^\mu}
\end{equation}
as well as 
\begin{equation} \label{comp2}
	C^A_\mu = v^A_\mu, 
		\quad 
	C_{A\mu}^\mu = -\frac{\partial L}{\partial y^A}, 
		\quad
	p_A^\mu = \frac{\partial L}{\partial v^A_\mu}.
\end{equation}
Consider now an integral section $\psi: X \rightarrow M$ of $\mathcal{X}_{n+1}$, locally represented by $\psi(x) = (x^\mu, y^A(x), y^A_\mu(x), p_A^\mu(x), p(x))$.  From the equations above, it follows that $\psi$ satisfies the following system of PDEs:
\[
	\frac{\partial y^A}{\partial x^\mu} = C^A_\mu, \quad
	\frac{\partial p_A^\nu}{\partial x^\mu} = C_{A\mu}^\nu, \quad 
	\frac{\partial p}{\partial x^\mu} = C_\mu
\]
where the local component functions $C^A_\mu$, $C_{A\mu}^\nu$ and $C_\mu$ are given by \eqref{comp1} and \eqref{comp2}.  After some simple manipulations, it can then be seen that $\psi$ satisfies the following equivalent set of equations:
\[
	\frac{\partial p_A^\mu}{\partial x^\mu} =  \frac{\partial L}{\partial y^A}, \quad
	\frac{\partial y^A}{\partial x^\mu} = v^A_\mu, \quad
	p_A^\mu = \frac{\partial L}{\partial v^A_\mu},
\]
together with
\[
	\frac{\partial}{\partial x^\mu} \left( p+p_A^\nu v^A_\nu - L \right) = 0.
\]
These are precisely the equations listed in the theorem statement.
\end{proof}

The equation \eqref{encons-LDS} can be integrated to $E=p+p_A^\nu v^A_\nu - L=\alpha$,  where $\alpha$ is a constant.  Without loss of generality we may impose the generalized energy constraint by setting $\alpha=0$ to obtain
 \[
 E=p+p_A^\nu v^A_\nu-L=0.
 \]
From this and \eqref{eqm-LDS}, we recover the covariant Legendre transformation
\[
p=L- \frac{\partial L}{\partial v^A_\mu} v^A_\nu, \quad p_A^\mu = \frac{\partial L}{\partial v^A_\mu}.
\]

\paragraph{Matrix Representation.}
Associated with the implicit Euler-Lagrange equations in \eqref{imp_EL}, 
the Lagrange-Dirac system given in equation \eqref{imp_EL} can be also denoted by
\[
\Omega_{M}^{\flat}\mathcal{X}=\mathbf{d}E, \qquad \mathbf{i}_{\mathcal{X}}(\pi_{XZ}^{\ast}\eta)=1,
\]
where a bundle map $\Omega_{M}^{\flat}: T^{n+1}M \to \Lambda^{1}M$ is defined by the pre-multisymplectic structure $\Omega_{M}$ by, for every $\mathcal{X} \in T^{n+1}M$,
\[
\mathbf{i}_{\mathcal{X}}\Omega_{M}=\Omega_{M}^{\flat}(\mathcal{X}).
\]
Then, one obtains a matrix representation as
\begin{equation*}\label{MatRepLDS}
\begin{split}
\left(
\begin{array}{c}
-\frac{\partial L}{\partial x^{\mu}} \vspace{2mm}\\
-\frac{\partial L}{\partial y^{A}} \vspace{2mm}\\
p^{\mu}_{A}-\frac{\partial L}{\partial v^{A}_{\mu}} \vspace{2mm}\\
v^{A}_{\mu} \vspace{2mm}\\
1
\end{array}
\right)
&=
\left(
\begin{array}{ccccc}
0 & 0& 0& 0&-1\\
0 & 0 & 0 & -1 & 0\\
0 & 0 & 0 & 0& 0\\
0 & 1 & 0 & 0 & 0\\
1 & 0 &0 & 0 & 0
\end{array}
\right)
\left(
\begin{array}{c}
\partial_{ x^{\mu}} x^{\mu}\vspace{1mm} \\
\partial_{ x^{\mu}} y^{A}  \vspace{1mm} \\
\partial_{ x^{\mu}} v^{A}_{\mu}  \vspace{1mm} \\
 \partial_{ x^{\mu}} p_{A}^{\mu}  \vspace{1mm} \\
\partial_{ x^{\mu}} p 
\end{array}
\right). 
\end{split}
\end{equation*}

Note that the above matrix representation for the Lagrange-Dirac formalism is related to the matrix representation of the multisymplectic formalism in \cite{Brid1997, MaPaSh1998}.

\begin{theorem}
The Lagrange-Dirac system $(\mathcal{X},E,D_{n+1})$ satisfies the condition of the {\bfi conservation of the generalized energy} as
\begin{equation} \label{econs}
	\mathcal{X} \prodint \mathbf{d} E = 0.
\end{equation}
In other words, along the solution of the Lagrange-Dirac system $(\mathcal{X},E,D_{n+1})$, namely, the integral manifold of $\mathcal {X}$, the generalized energy $E$ is constant. 
\end{theorem}

\begin{proof}
The proof relies on lemma~\ref{lemma:decomp} for decomposable vector fields.  Using the field equations, the inner product on the left-hand side of \eqref{econs} can be written as 
\[
	\mathcal{X} \prodint \mathbf{d} E = \mathcal{X} \prodint (\mathcal{X} \intprod \Omega_{M}) = 0,
\]
according to corollary~\ref{cor:prodprod}.
\end{proof}

To see why the previous theorem implies energy conservation, decompose the multivector field $\mathcal{X}$ as in \eqref{multvectf}.  Using lemma~\ref{lemma:decomp}, the interior product can then be written as 
\[
	\mathcal{X} \prodint \mathbf{d} E = 
	\sum_{\mu = 1}^k (-1)^{\mu + 1} \left<\mathcal{X}_\mu, 
		\mathbf{d} E \right> 
				\hat{\mathcal{X}}_\mu
\]
where $\hat{\mathcal{X}}_\mu$ is the $n$-multivector field obtained by deleting $\mathcal{X}_\mu$ from $\mathcal{X}$, i.e. 
\[
		\hat{\mathcal{X}}_\mu = \bigwedge_{\stackrel{\lambda = 1}{\lambda \ne \mu}}^{n+1} \mathcal{X}_\lambda.
\]
Since the multivector fields $\hat{\mathcal{X}}_\mu$ are linearly independent, the energy conservation equation \eqref{econs} then implies that $\left<\mathcal{X}_\mu, \mathbf{d} E \right> = 0$ for $\mu = 1, \ldots, n + 1$.  In other words, the function $E$ is constant on the integral sections of $\mathcal{X}$.

\section{Nonholonomic Lagrange-Dirac Field Theories}

In this section, we develop the idea of {\it nonholonomic Lagrange-Dirac field theories}. First, we review the theory of mechanical systems in the presence of nonholonomic constraints, in order to get an idea of the corresponding structures for field theory.  Secondly, we define a {\it multi-Dirac structure with nonholonomic constraints} by analogy with the induced Dirac structure in nonholonomic mechanics as in \cite{YoMa2006a, YoMa2006b}.  We focus especially on the case of \emph{affine constraints}, in which case the induced multi-Dirac structure takes on a particularly elegant form.  Then, we show how the nonholonomic Lagrange-Dirac field theory can be developed in the context of the induced nonholonomic multi-Dirac structure.

\paragraph{Lagrange-Dirac Systems in Nonholonomic Mechanics.}

Before going into details on nonholonomic Lagrange-Dirac field theories, let us make a brief review on nonholonomic mechanics in the context of Dirac structures. 
 
Let $Q$ be a configuration manifold and let $M=TQ \oplus T^{\ast}Q$ be the Pontryagin bundle over $Q$. Recall from \cite{YoMa2006a} that a Dirac structure $D_{\Delta_{Q}}$ on  $Z=T^{\ast}Q$  can be induced from a constraint distribution $\Delta_{Q}$ on $Q$, where $\Delta_{Q}$ is not integrable in general, namely, the constraint is {\it nonholonomic}. Let $L$ be a Lagrangian on $TQ$, possibly degenerate and let $E(q,v,p)=\left< p, v\right>-L(q,v)$ be the generalized energy on $M$.  Furthermore, we define a partial vector field in this context to be a map $X : M \rightarrow TM$ such that 
\[
	T \pi_{TQ, M} \circ X = 0,
\]
where $\pi_{TQ, M}: M \rightarrow TQ$ is the projection onto the first factor.  In local coordinates, a partial vector field $X$ can be written as $X(q, v, p) = (q, v, p, \dot{q}, 0, \dot{p})$.

We can develop the Lagrange-Dirac dynamical system in the context of the {\bfi induced Dirac structure on Pontryagin bundle} $M$ (see \cite{YoMa2006b} and \cite{CMRY2010}). To do this, let $\Delta_{M} \subset TM$  be the distribution on $M$ defined by $\Delta_{M} = (T\pi_{QM})^{-1}(\Delta_{Q}),$ where $\pi_{QM} : M \rightarrow Q$ is the Pontryagin bundle projection, which in coordinates is denoted by $\pi_{QM}(q,v,p) = q$. Since a pre-symplectic form $\Omega_{M}$ is defined by using the projection $\pi_{ZM}: M \to Z$ as 
\[
\Omega_{M}=\pi_{ZM}^{\ast}\Omega,
\]
we can define an induced Dirac structure $D_{\Delta_{M}} \subset TM \oplus T^\ast M$ on
$M$ by, for $m \in M$,
\begin{align*}\label{induceddirac_M}
D_{M}(m)  =\{\, (w_{m}, \beta_{m})  &
\in T_{m}M \times T^{\ast}_{m}M \mid w_m \in
\Delta_{M}(m),  \nonumber \\
  & \qquad \qquad \; \mbox{and} \;\;
\beta_{m}-\Omega_{M}^{\flat}(m) \cdot w_{m}
\in \Delta^{\circ}_{M}(m) \, \}.
\end{align*}
In coordinates $(q,v,p)$ for $m \in M$, one has
\begin{equation*}\label{bardiracstructure}
\begin{split}
D_{\Delta_{M}}(q,v,p)
=
\{
((\dot{q}, \dot{v}, \dot{p}), (\alpha, \gamma, \beta)) \mid \dot{q} \in
\Delta(q),
 \alpha + \dot{p} \in \Delta^\circ (q),\,\beta = \dot{q}, \gamma = 0
\},
\end{split}
\end{equation*}
where the constraint set $\Delta_Q$  defines a subspace of $TQ$,
which is expressed in a local trivialization $U \subset Q$ by $\Delta(q) \subset \mathbb{R}^n$ at each point
$q \in U $.

A nonholonomic Lagrange-Dirac dynamical system is then specified as
\begin{equation*}\label{ImpLagSysZ}
(X, \mathbf{d}E) \in D_{\Delta_{M}},
\end{equation*} 
from which we obtain
\begin{equation}\label{int-ImpLagSys}
\mathbf{d}E-\Omega_{M}^{\flat} \cdot X \in \Delta_{M}^{\circ} \quad \mbox{and} \quad X \in \Delta_{M}. 
\end{equation}
In local coordinates, the nonholonomic Lagrange-Dirac dynamical system in \eqref{int-ImpLagSys} is represented as
\begin{equation}\label{Local-ImpLagSys}
\dot{p}
-
\frac{\partial L}{\partial q} 
\in 
\Delta^\circ (q),
\quad
\dot{q}
=
v,
\quad
p -\frac{\partial L}{\partial v} 
= 
0,
\quad
\dot{q}
\in
\Delta (q).
\end{equation}

\paragraph{Nonholonomic Constraints for Field Theories.}

We now show how multi-Dirac structures make an appearance in the theory of classical field theories with nonholonomic constraints by analogy with nonholonomic mechanics.  For simplicity, we restrict ourselves to the case of constraints that are affine in the multi-velocities --- the general case may be dealt with in a similar way as in \cite{Vank2005}.

\paragraph{Affine Nonholonomic Constraints.}
Consider a fiber bundle $\pi_{XY}: Y \rightarrow X$, where $X$ is an oriented manifold with $\dim X = n + 1$ and with a fixed volume form $\eta$, locally given by $\eta=d^{n+1}x$. Let $\Delta_{Y}$ be a distribution on $Y$.  By following \cite{Krupkova05}, we say that $\Delta_{Y}$ is {\bfi weakly horizontal} when there exists a distribution $W$ contained in the vertical bundle $VY$ such that 
\[
	\Delta_{Y} \oplus W = TY.
\] 
When $W$ is the whole of $VY$, the distribution $\Delta_{Y}$ is nothing but the horizontal bundle of an Ehresmann connection on $Y$.  More information on weakly horizontal distributions and their role in nonholonomic field theory can be found in \cite{KrupkovaVolny2}.  In coordinates $(x^{\mu}, y^{A})$ on $Y$, the annihilator $\Delta_{Y}^{\circ}$ is spanned by $k$ local one-forms $\varphi^{\alpha}$ that are given by
\[
\varphi^{\alpha}(x^{\mu}, y^{A})=A^{\alpha}_{A}(x, y) dy^{A} +A^{\alpha}_{\mu}(x, y) dx^{\mu}, \quad \alpha=1,...,k <N,
\]
where the rank of the matrix $A^{\alpha}_{A}(x, y)$ is $k$.  

A weakly horizontal distribution $\Delta_{Y}$ defines an affine subbundle $\mathcal{C}$ of $J^1 Y$ as follows.  We recall that an element $\gamma=j^{1}_{x}\phi \in J^{1}Y$ may be viewed as an injective linear map $\gamma=T_{x}\phi: T_{x}X \to T_{\phi(x)}Y$ such that $T \pi_{XY} \circ \gamma = \mathrm{Id}_{T_x X}$.  The elements of $\mathcal{C}$ are then the one-jets $\gamma$ taking values in $\Delta_{Y}$:
\[
	\gamma \in \mathcal{C} \quad \text{if} \quad
	\mathrm{Im}\, \gamma \subset \Delta_{Y}(y).	
\]
Locally, $\mathcal{C}$ can be characterized as follows.  Since $\gamma$ is given in coordinates by
\[
\gamma = d x^\mu \otimes \left( \frac{\partial}{\partial x^\mu} + v^A_\mu \frac{\partial}{\partial y^A} \right),
\]
we have that $\gamma$ takes values in $\Delta_Y$ iff 
\begin{equation} \label{affconstr}
	\psi^\alpha_\mu(x, y, v) \equiv 
	A^\alpha_\mu(x, y) + A^\alpha_A(x, y) v^A_\mu = 0.
\end{equation}
In other words, $\mathcal{C}$ is locally determined by the vanishing of the $k$ independent \emph{affine functions} $\psi^\alpha_\mu$.

From now on, we consider affine constraints of the form \eqref{affconstr} which come from a linear distribution $\Delta_{Y}$.  We say that these constraints are {\bfi nonholonomic} if $\Delta_{Y}$ is not integrable.  In order to incorporate the nonholonomic constraints into the context of multi-Dirac structures on the Pontryagin bundle $M=J^{1}Y \oplus Z$ over $Y$, we introduce a distribution $\Delta_{M}$ along $\mathcal{C}$ as follows: 
\[
	\Delta_M = (T (\pi_{YM})_{|\mathcal{C}})^{-1}(\Delta_Y).
\]
Note that $\Delta_M$ is not a distribution on the entire Pontryagin bundle, but only on the affine submanifold $\mathcal{C}$.  We denote the annihilator of $\Delta_M$ by $\Delta_M^\circ$.  Note that $\Delta_M^\circ$ is locally spanned by the following $k$ linear independent forms 
\[
	\tilde{\varphi}^\alpha = A^\alpha_A ( dy^A - v^A_\mu dx^\mu ). 
\]
Another way of defining $\tilde{\varphi}^\alpha$, which also works for nonlinear constraints, is by means of the \emph{vertical endomorphism} $S$ on $J^1 Y$.   We refer to \cite{VCdeLMD2005} for further details.

\paragraph{Nonholonomic Multi-Dirac Structures.}

Now, we show how a {\it multi-Dirac structure} on $M$ can be induced from the nonholonomic constraint distribution $\Delta_{M}$.

\newtext{
\begin{proposition} \label{prop: dirac-const}
For $r = 1, \ldots, n+1$, we define the following subbundles $D_{\Delta_M, r} \subset P_r = T^r M \times \Lambda^{n + 2 - r}(M)$, given by 
\begin{equation}\label{multi-dirac-const}
\begin{split}
D_{\Delta_{M},r}&=\left\{ (\mathcal{X}, \Sigma)  \in P_{r}  \mid  
 \;\mathbf{i}_{\mathcal{X}} \Omega_{M}-\Sigma \in 
 	\mbox{$\bigwedge^{n + 2 - r}$} (\Delta_M^\circ), \;\; 
	\mathcal{X} \in \Delta_{M} \mbox{$\bigwedge$} T^{r - 1} M \right\}.
\end{split}
\end{equation}
Then $D_{\Delta_M} = D_{\Delta_M, 1}, \ldots, D_{\Delta_M, n + 1}$ is a multi-Dirac structure of degree $n + 1$.  Namely, the sequence of bundles $D_{\Delta_{M}, r}$ satisfies the $(n + 1)$-Lagrangian property
\begin{equation} \label{isononh}
D_{\Delta_{M},r}=(D_{\Delta_{M},s})^{\perp,r}
\end{equation}
for all $r,s=1,...,n+1; r+s \le n+2$. 
\end{proposition}

\begin{proof}
We have to show that \eqref{isononh} holds for all indices $r, s$ with $r + s \le n + 2$.  Let us first check that 
\begin{equation} \label{weakinclusion}
	D_{\Delta_{M},r} \subset (D_{\Delta_{M},s})^{\perp,r}.
\end{equation}
Let $(\mathcal{X},\Sigma) \in D_{\Delta_{M},r}$
and $(\bar{\mathcal{X}},\bar{\Sigma}) \in D_{\Delta_{M},s}$.  By definition, there exist $\alpha \in \mbox{$\bigwedge^{n + 2 - r}$} (\Delta_M^\circ)$ and $\bar{\alpha} \in \mbox{$\bigwedge^{n + 2 - s}$} (\Delta_M^\circ)$ such that 
\[
	\mathbf{i}_{\mathcal{X}} \Omega_M - \Sigma = \alpha 
	\quad \text{and} \quad 
	\mathbf{i}_{\bar{\mathcal{X}}} \Omega_M - \bar{\Sigma} = \bar{\alpha}.
\]
Then, it follows that
\[
\begin{split}
&\left<\left< (\mathcal{X},\Sigma),(\bar{\mathcal{X}},\bar{\Sigma})\right>\right>_{-}\\
&\hspace{1cm}=\frac{1}{2}\left\{\mathbf{i}_{\bar{\mathcal{X}}} \Sigma +(-1)^{rs+1}\mathbf{i}_{\mathcal{X}} \bar{\Sigma}\right\}\\
&\hspace{1cm}=\frac{1}{2}\left\{\mathbf{i}_{\bar{\mathcal{X}}}\left( \mathbf{i}_{\mathcal{X}} \Omega_{M} +\alpha \right)+(-1)^{rs+1} \mathbf{i}_{\mathcal{X}}\left( \mathbf{i}_{\bar{\mathcal{X}}}\Omega_{M} + \bar{\alpha} \right)\right\}
\\
&\hspace{1cm}=0,
\end{split}
\]
since $\mathbf{i}_{\bar{\mathcal{X}}}\mathbf{i}_{\mathcal{X}}\Omega_{M}=(-1)^{rs} \mathbf{i}_{\mathcal{X}}\mathbf{i}_{\bar{\mathcal{X}}}\Omega_{M}$, together with $\mathbf{i}_{\bar{\mathcal{X}}}\alpha=0$ and $\mathbf{i}_{\bar{\mathcal{X}}} \bar{\alpha}$.  We conclude that \eqref{weakinclusion} holds.

Next, let us show that the reverse inclusion 
\begin{equation} \label{revinclusion}
	(D_{\Delta_{M},s})^{\perp,r} \subset D_{\Delta_{M},r}
\end{equation}
holds.  Let $(\bar{\mathcal{X}}, \bar{\Sigma}) \in (D_{\Delta_{M},s})^{\perp,r}$. By definition of $(D_{\Delta_{M},s})^{\perp,r}$,
\[
\mathbf{i}_{\bar{\mathcal{X}}} \Sigma + (-1)^{rs+1}\mathbf{i}_{\mathcal{X}} \bar{\Sigma}=0
\]
for all $(\mathcal{X},\Sigma)  \in D_{\Delta_{M},s}$, i.e. 
$\mathcal{X} \in \Delta_M \wedge T^{s-1}Z$ such that 
\[
\mathbf{i}_{\mathcal{X}} \omega-\Sigma = \alpha \in \mbox{$\bigwedge^{n + 2 - s}$} (\Delta_M^\circ).
\]  
It follows that 
\[
\begin{split}
\mathbf{i}_{\bar{\mathcal{X}}} \Sigma + (-1)^{rs+1}\mathbf{i}_{\mathcal{X}} \bar{\Sigma} &=\mathbf{i}_{\bar{\mathcal{X}}} (\mathbf{i}_{\mathcal{X}} \Omega_{M} + \alpha)+ (-1)^{rs+1}\mathbf{i}_{\mathcal{X}} \bar{\Sigma} \\
&=(-1)^{rs}\mathbf{i}_{\mathcal{X}} \left\{   \mathbf{i}_{\bar{\mathcal{X}}} \Omega_{M} - \bar{\Sigma}\right\} + \mathbf{i}_{\bar{\mathcal{X}}} \alpha
\end{split}
\]
has to vanish for all $\mathcal{X} \in \Delta_M \wedge T^{s - 1}M$. We first choose $\mathcal{X} = 0$, so that the expression between brackets vanishes.  In this case, the condition becomes $\mathbf{i}_{\bar{\mathcal{X}}} \alpha = 0$ for all $\alpha \in \mbox{$\bigwedge^{n + 2 - s}$} (\Delta_M^\circ)$ and hence $\bar{\mathcal{X}} \in \Delta_M \wedge T^{r - 1}M$ by lemma~\ref{lemma:vectorfield}.  Secondly, for arbitrary $\mathcal{X} \in \Delta_M \wedge T^{s - 1}M$, we then have the following condition: 
\[
	\mathbf{i}_{\mathcal{X}} \left\{   \mathbf{i}_{\bar{\mathcal{X}}} \Omega_{M} - \bar{\Sigma}\right\} = 0, 
\]
so that the expression between brackets is an element of $\bigwedge^{n + 2 - r} (\Delta_M^\circ)$ by lemma~\ref{lemma:form}:
\[
\mathbf{i}_{{\bar{\mathcal{X}}}}\Omega_{M}-\bar{\Sigma} \in \mbox{$\bigwedge^{n + 2 - r}$}(\Delta_M^\circ).
\]
This concludes the proof of the reverse inclusion \eqref{revinclusion}.  We conclude that $D_{\Delta_M} = D_{\Delta_M, 1}, \ldots, D_{\Delta_M, n + 1}$ is a multi-Dirac structure of degree $n + 1$
\end{proof}
} % color red

We refer to the multi-Dirac structure $D_{\Delta_M} = D_{\Delta_M, 1}, \ldots, D_{\Delta_M, n + 1}$ in \eqref{multi-dirac-const} as the {\bfi nonholonomic multi-Dirac structure} induced by the distribution $\Delta_{M} \subset TM$.  Needless to say that $D_{\Delta_{M},r}$ is a field-theoretic analogue of the induced Dirac structure from a given constraint distribution in mechanics. Furthermore, we note that, as in the case of field theories without constraints, only the highest component $D_{\Delta_M, n + 1}$ plays a role in the formulation of the nonholonomic field equations.  \newtext{  \cite{Zambon2010} considers an analogue of nonholonomic multi-Dirac structures in the context of higher-order Dirac structures.}

\paragraph{Nonholonomic Lagrange-Dirac Field Theories.}

Consider a nonholonomic multi-Dirac structure $D_{\Delta_M} = D_{\Delta_M, 1}, \ldots, D_{\Delta_M, n + 1}$.  In the following, only the component $D_{\Delta_M, n + 1} \subset P_{n+1}=T^{n+1}M \oplus \Lambda^{1}M$ will be used. Explicitly, this component is given by
\newtext{ \begin{equation*}\label{ind-multi-Dirac-n+1}
\begin{split}
D_{\Delta_{M},n+1}&=\left\{ (\mathcal{X}, \Sigma)  \in P_{n+1}  \mid  
 \;\mathbf{i}_{\mathcal{X}} \omega-\Sigma \in \Delta_{M}^{\circ}, \;\; \mathcal{X} \in \Delta_{M} \wedge T^n M \right\}.
\end{split}
\end{equation*} }

Now, let $\mathcal{L} = L\eta : J^1 Y \rightarrow \Lambda^{n + 1}(X)$ be a Lagrangian density with associated generalized energy $E = p + p_A^\mu v^A_\mu - L$, and let $\mathcal{X}$ be a \emph{partial vector field} on $M$.  Recall that $\mathcal{X}$ can locally be written as 
\begin{equation} \label{partvf}
	\mathcal{X} 
	= \bigwedge_{\mu = 1}^{n + 1}
		\left( \frac{\partial}{\partial x^\mu}
			+ C^A_\mu \frac{\partial}{\partial y^A}
			+ C_{A\mu}^\nu \frac{\partial}{\partial p_A^\nu}
			+ C_\mu \frac{\partial}{\partial p} \right), \quad
         \mathbf{i}_{\mathcal{X}}(\pi_{XM}^{\ast}\eta)=1.
\end{equation}

\newtext{
\begin{definition}\label{def: NHLagDiracFT}
A {\bfi nonholonomic Lagrange-Dirac system for field theories with affine constraints} is a quadruple $(\mathcal{X},E, \Delta_M, D_{\Delta_{M}})$, where $\mathcal{X}$ is a partial multivector field, $E$ is a generalized energy, $\Delta_M$ is a distribution on $M$, and $D_{\Delta_M} = D_{\Delta_M, 1}, \ldots, D_{\Delta_M, n + 1}$ is a nonholonomic multi-Dirac structure, which satisfies
\begin{equation}\label{cond_NHLagDiracFT}
(\mathcal{X}, (-1)^{n+2}\mathbf{d}E) \in D_{\Delta_{M},n+1}.
\end{equation}
and
\begin{equation} \label{conditiondelta}
	\mathcal{X} \in \mbox{$\bigwedge^{n + 1}$} \Delta.
\end{equation}
\end{definition}

Note that the condition \eqref{conditiondelta} in this definition does not appear in the case of mechanical systems with nonholonomic constraints.  From \eqref{cond_NHLagDiracFT} we have that $\mathcal{X} \in \Delta \wedge T^n M$, and for mechanical systems, where $n = 0$, this implies that $\mathcal{X} \in \Delta$, i.e. \eqref{conditiondelta} holds automatically.  For general field theories this is no longer the case, and \eqref{conditiondelta} arises as an independent condition.  
}

\begin{theorem} \label{thm:LDeqnh} The nonholonomic Lagrange-Dirac system $(\mathcal{X},E,D_{\Delta_{M},n+1})$ in equation  \eqref{cond_NHLagDiracFT} induces {\bfi nonholonomic Lagrange-Dirac equations for field theories}:
\begin{equation} \label{LagDiracEq}
\newtext{ \mathcal{X} \in \mbox{$\bigwedge^{n + 1}$} \Delta}, \quad
\mathbf{i}_{\mathcal{X}}\Omega_{M}-(-1)^{n+2}\mathbf{d}E \in \Delta_{M}^{\circ}, \qquad \mathbf{i}_{\mathcal{X}}(\pi_{XM}^{\ast}\eta)=1,
\end{equation}
which can be written in local coordinates  
as
\begin{equation} \label{eqm}
	\frac{\partial p_A^\mu}{\partial x^\mu} -  \frac{\partial L}{\partial y^A} =\lambda_{\alpha}A^{\alpha}_{A}, \quad
	\frac{\partial y^A}{\partial x^\mu} = v^A_\mu, \quad
	p_A^\mu = \frac{\partial L}{\partial v^A_\mu}
\end{equation}
and 
\begin{equation} \label{energy-const-NH}
\frac{\partial}{\partial x^\mu} (p + p_A^\mu v^A_\mu - L) = 0,
\end{equation}
together with nonholonomic affine constraints
\[
A^{\alpha}_{A}v^{A}_{\mu}+A^{\alpha}_{\mu}=0,
\]
where $\lambda_{\alpha}$ are Lagrange multipliers. It follows from equation \eqref{energy-const-NH} that imposing the generalized energy constraint $E=0$ reads
\[
p=L-p^{\mu}_{A}v^{A}_{\mu}.
\]
\end{theorem}
\begin{proof}
By the direct computations using local coordinates, we can easily check that equation \eqref{LagDiracEq} yields 
\eqref{eqm} 
as well as
\[
	\frac{\partial y^A}{\partial x^\mu} \frac{\partial p_A^\nu}{\partial x^\nu} - 
	\frac{\partial y^A}{\partial x^\nu} \frac{\partial p_A^\nu}{\partial x^\mu} - 
	\frac{\partial p}{\partial x^\mu} = -\frac{\partial L}{\partial x^\mu}
	- \lambda_\alpha A^\alpha_A v^A_\mu,
\]
which is equation \eqref{energy-const-NH}.
\end{proof}

\paragraph{Integrability of Nonholonomic Constraints.}  We shall investigate the integrability of the nonholonomic multi-Dirac structure.  Recall that the constraint constraints in this paper are given by an affine subbundle $\mathcal{C}$ of $J^1 Y$ which induces the constraint distribution $\Delta_{M}$ on $M$ and also that the constraints are said to be nonholonomic if $\Delta_{M}$ is not integrable.

Consider the nonholonomic multi-Dirac structure $D_{\Delta_M} = D_{\Delta_M, 1}, \ldots, D_{\Delta_M, n + 1}$ of degree $n + 1$ defined previously, and endowed with the multi-Courant bracket given in \eqref{multicourant}.  Then, the following proposition shows that integrability is equivalent to the space of section of $D_{\Delta_{M}}$ being closed under the multi-Courant bracket.

\begin{proposition}
The nonholonomic multi-Dirac structure $D_{\Delta_M} = D_{\Delta_M, 1}, \ldots, D_{\Delta_M, n + 1}$ is \emph{integrable}  if and only if the following conditions are satisfied:
\begin{itemize}
\item
The distribution $\Delta_{M}$ is integrable; namely, the condition 
$
[\mathcal{X}, \bar{\mathcal{X}}] \in \Delta_{M}
$
holds for $\mathcal{X}, \bar{\mathcal{X}} \in \Delta_{M}$.

\item
The pre-multisymplectic structure $\Omega_{M}$ \newtext{satisfies $\mathbf{i}_{X \wedge Y \wedge Z}
\mathbf{d}\Omega_{M}=0$ for any three $X, Y, Z \in \Delta_M$}.
\end{itemize}
\end{proposition}
\begin{proof}
Recall the multi-Courant bracket on the space of sections of $P_{r} \times P_{s}$, namely,
\[
\left[\!\left[ \cdot, \cdot \right]\!\right]_{r,s}: \Gamma(P_{r}) \times \Gamma(P_{s}) \to \Gamma(P_{r+s-1})
\]
is given by
\begin{align*}
[\![ (\mathcal{X}, \Sigma),(\bar{\mathcal{X}}, \bar{\Sigma})]\!]_{r,s}  
&=
\left( [\mathcal{X},\bar{\mathcal{X}}], \; \pounds_{\mathcal{X}}\bar{\Sigma}-(-1)^{(r-1)(s-1)}\pounds_{\bar{\mathcal{X}}}\Sigma+\frac{(-1)}{2}^{r}\mathbf{d}
\left<\!\left<(\mathcal{X}, \Sigma),(\bar{\mathcal{X}}, \bar{\Sigma})\right>\!\right>_{+}
\right)\\
&=  \left([\mathcal{X}, \bar{\mathcal{X}}], \pounds_{\mathcal{X}}\bar{\Sigma}  - (-1)^{(r-1)(s-1)}\pounds_{\bar{\mathcal{X}}} \Sigma +
 \frac{(-1)^{r}}{2} \mathbf{d} (\mathbf{i}_{\mathcal{X}} \bar{\Sigma} +(-1)^{rs}\mathbf{i}_{\bar{\mathcal{X}}}\Sigma)\right). 
\end{align*}
For $(\mathcal{X},\Sigma) \in D_{\Delta_{M},r}$ and $(\bar{\mathcal{X}}, \bar{\Sigma}) \in D_{\Delta_{M},s}$,
one has
\begin{align*}
[\![ (\mathcal{X}, \Sigma),(\bar{\mathcal{X}}, \bar{\Sigma})]\!]_{r,s} & =  ([\mathcal{X}, \bar{\mathcal{X}}], \pounds_{\mathcal{X}}\bar{\Sigma}  - (-1)^{(r-1)(s-1)}\pounds_{\bar{\mathcal{X}}} \Sigma +
 \frac{(-1)^{r}}{2} \mathbf{d} (\mathbf{i}_{\mathcal{X}} \bar{\Sigma} +(-1)^{rs}\mathbf{i}_{\bar{\mathcal{X}}}\Sigma)
) \\
& = ([\mathcal{X}, \bar{\mathcal{X}}], \pounds_{\mathcal{X}} \mathbf{i}_{\bar{\mathcal{X}}} \Omega_{M} - (-1)^{(r-1)(s-1)} \mathbf{i}_{\bar{\mathcal{X}}} \pounds_{\mathcal{X}} \Omega_{M} +  (-1)^{r} \mathbf{i}_{\mathcal{X}} \mathbf{i}_{\bar{\mathcal{X}}} \mathbf{d}\Omega_{M} ) \\
& = ([\mathcal{X}, \bar{\mathcal{X}}],  \mathbf{i}_{[\mathcal{X}, \bar{\mathcal{X}}]}\Omega_{M} +  (-1)^{r} \mathbf{i}_{\mathcal{X}} \mathbf{i}_{\bar{\mathcal{X}}}  \mathbf{d}\Omega_{M} ).
\end{align*}
\newtext{
It follows that the multi-Courant bracket is closed iff 
\[
	[\mathcal{X}, \bar{\mathcal{X}}]  \in \Delta_M \wedge T^{r + s - 2} M
\quad \text{and} \quad 
	\mathbf{i}_{\mathcal{X} \wedge \bar{\mathcal{X}}} \, \mathbf{d} \Omega_M \in \mbox{$\bigwedge$}^{n + 3 - (r + s)}(\Delta_M^\circ).  
\]
This is equivalent to the conditions in the statement of the proposition.
}
\end{proof}

\section{Lagrange-d'Alembert-Pontryagin Principle for Field Theories}

Let us show how {\it nonholonomic Lagrange-Dirac field equations} can be formulated by the {\it Lagrange-d'Alembert-Pontryagin principle}, which is an extension of the Hamilton-Pontryagin principle to the case in which nonholonomic constraints are given. Before going into this, let us review Lagrange-d'Alembert-Pontryagin principle in mechanics.

\paragraph{Lagrange-d'Alembert-Pontryagin Principle in Mechanics.}

We consider again the case of a mechanical system first, formulated on a jet bundle $\pi: Y \rightarrow X$ where $X = \mathbb{R}$.  The coordinate on $X$ is denoted $t$ and represents time.  Let $\Theta_{M}=\pi_{ZM}^{\ast}\Theta$ be the one-form on $M$ that is induced from the canonical one-form $\Theta$ on $Z$.  We can also derive the Lagrange--d'Alembert--Pontryagin  by using the {\bfi Lagrange-d'Alembert-Pontryagin principle} (see \cite{YoMa2006b}), which is given by the stationary condition for the action integral 
$\mathfrak{F}(m)$ of a  curve $m(t)$ in $M$ as

\begin{equation}\label{LDAP_int}
\delta \mathfrak{F}(m)=\delta \int_{a}^{b} \left\{ \Theta_{M}(m(t)) \cdot \dot{m}(t)-E(m(t)) \right\}dt=0
\end{equation}
for all $\delta{m}(t) \in \Delta_{M}(m(t))$ and with the constraint  $\dot{x}(t) \in \Delta_{M}(m(t))$ and the endpoint fixed conditions $T\pi_{QM}(\delta{m}(a))=T\pi_{QM}(\delta{m}(b))=0$.  

By direct computations, it follows that
\begin{equation}\label{Int-LDAP}
\begin{split}
\delta \mathfrak{F}(m)=\frac{d}{d\lambda}\mathfrak{F}(m_{\lambda})\biggr|_{\lambda=0}&=\frac{d}{d\lambda}\int_{a}^{b} \left\{ \Theta_{M}(m_{\lambda}(t)) \cdot \dot{m}_{\lambda}(t)-E(m_{\lambda}(t)) \right\}dt\\
&=\int_{a}^{b} \left\{ \mathbf{d}\Theta_{M}(m)(\delta{m}, \dot{m})-\left<\mathbf{d}E(m), \delta{m} \right>\right\}dt\\
&=\int_{a}^{b} \left\{ \Omega_{M}(m)(\dot{m}, \delta{m})-\left<\mathbf{d}E(m), \delta{m} \right>\right\}dt\\
&=\int_{a}^{b} \left<\Omega_{M}^{\flat}(m)\cdot \dot{m}-\mathbf{d}E(m), \delta{m} \right>dt\\
&=0
\end{split}
\end{equation}
for all $\delta{m}(t) \in \Delta_{M}(m(t))$ and with the constraint  $\dot{m}(t) \in \Delta_{M}(m(t))$ and the endpoint fixed conditions $T\pi_{QM}(\delta{m}(a))=T\pi_{QM}(\delta{m}(b))=0$. 
\medskip

Thus, we obtain the {\bfi Lagrange--d'Alembert--Pontryagin  equations} as
\begin{equation}
\Omega_{M}(m)^{\flat}\cdot \dot{m}-\mathbf{d}E(m) \in \Delta_{M}^{\circ}(m) \quad \mbox{and} \quad \dot{m} \in \Delta_{M}(m). 
\end{equation}
In local coordinates $(q,v,p)$ for $m \in M$, it follows that
\begin{equation}\label{LDAP}
\begin{split}
\delta \mathfrak{F}(q,v,p)&=\delta \int_{a}^{b} \left\{  \left< p(t), \dot{q}(t)\right> -E(q(t),v(t),p(t))\right\} dt\\
&=\delta \int_{a}^{b} \left\{ L(q(t),v(t)) + \left< p(t), \dot{q}(t)-v(t)\right> \right\} dt=0
\end{split}
\end{equation}
for all variations $\delta{q}(t) \in \Delta_{Q}(q(t))$ and with the constraint $\dot{q}(t) \in \Delta_{Q}(q(t))$ and the endpoint fixed conditions $\delta{q}(a)=\delta{q}(b)=0$. 
It follows from equation \eqref{LDAP} that we can obtain the Lagrange--d'Alembert--Pontryagin  equations as in \eqref{Local-ImpLagSys}.

\paragraph{Nonholonomic Field Theories and Bundles of Reaction Forces.}
Now, recall the geometric setting of the nonholonomic Lagrange-Dirac field theories. Namely, let $\pi_{XY}: Y \rightarrow X$ be a fiber bundle, where $X$ is an oriented manifold with $\dim X = n + 1$ and with a fixed volume form $\eta$, locally given by $\eta=d^{n+1}x$. Let $M=J^{1}Y \oplus Z$ be the Pontryagin bundle over Y.
Let $\Delta_{Y}$ be a distribution on $Y$. As before, suppose that the annihilator $\Delta_{Y}^{\circ}$ is spanned by
$k$ one-forms $\varphi^{\alpha}$ that are given in coordinates $(x^{\mu}, y^{A})$ for $Y$ by
\[
\varphi^{\alpha}(x^{\mu}, y^{A})=A^{\alpha}_{A}dy^{A} +A^{\alpha}_{\mu}dx^{\mu}, \quad \alpha=1,...,k <N,
\]
and $\Delta_{Y}$ defines an affine subbundle $\mathcal{C}$ of $J^1 Y$ which is given in coordinates $(x^{\mu},y^{A},v^{A}_{\mu})$ for $J^{1}Y$ by 
\[
\begin{split}
\mathcal{C}:=&\{ (x^{\mu}, y^{A},v^{A}_{\mu}) \in J^{1}Y \; \mid \; \varphi^\alpha_\mu=A^\alpha_A(x, y) y^A_\mu +  A^\alpha_\mu(x, y) = 0 ,\\
&\hspace{6cm} \alpha=1,...,k;\; \mu=1,...,n+1 \}.
\end{split}
\]
Further, recall that the distribution $\Delta_M$ on $M$ is given by using $\pi_{YM}: M \to Y$ as
\begin{equation*}\label{Dist-M}
\Delta_M = (T (\pi_{YM})_{|\mathcal{C}})^{-1}(\Delta_Y),
\end{equation*}
namely,
\[
\Delta_{M}:=\left\{ X \in TM \;\mid \; \left< \tilde{\varphi}^{\alpha},  X \right> =0,\quad \alpha=1,...,k \right\},
\]
where $\tilde{\varphi}^{\alpha}$ are $k$ one-forms on $M$ given by
\[
\tilde{\varphi}^{\alpha}=A^{\alpha}_{A}(dy^{A} -v^{A}_{\mu}dx^{\mu}),
\]
which span the annihilator of $\Delta_{M}$. Moreover, the $r$-fold exterior power of $\Delta_{M}$ is defined by
\[
\Delta_{M}^{r}= \bigwedge_{}^{r}\Delta_{M},\quad r=1,...,n+1. 
\]

We now introduce a closely related concept, namely the {\bfi bundle of reaction forces} $F \subset \Lambda^{n+1}(M)$.  This bundle is spanned by the $k(n+1)$ $(n + 1)$-forms $\Phi^\alpha_\mu$ along $\mathcal{C}$ given by 
\[
\begin{split}
\Phi^\alpha_\mu &= \tilde{\varphi}^{\alpha}\wedge d^{n}x_\mu \\
&=A^{\alpha}_{A}(dy^{A} -v^{A}_{\nu}dx^{\nu})\wedge d^{n}x_\mu.
\end{split}
\]
An intrinsic way of defining the bundle $F$, as well as further properties, can be found in \cite{VankMD2008}.

\begin{definition}
A vertical variation $\mathcal{V}_{M}$ of a section $\psi$ of $\pi_{XM}: M \to X$ that takes values in $\Delta_{M} \subset TM$ defined on an open subset $U$ with compact closure is {\it admissible} if 
\begin{equation} \label{admvariations}
\psi^{\ast}(\mathcal{V}_{M}\! \text{\Large$\lrcorner$} \; \Phi)=0 \quad \mbox{for all} \quad \Phi \in F.
\end{equation}
\end{definition}

\vspace{-0.5cm}
\paragraph{Lagrange-d'Alembert-Pontryagin Principle for Field Theories.}
Consider a Lagrangian density $\mathcal{L}$ and with its associated energy density $\mathcal{E}$.  
Recall from Definition \ref{def:HPActionFunct} that the Hamilton-Pontryagin action functional is given by 
\begin{equation}\label{HPAction}
	S(\gamma, z) = \int_{U}  \psi^{\ast}(\Theta_{M} -\mathcal{E}),
\end{equation}
where $U$ is an open subset of $X$ with compact closure and $\psi=(\gamma, z)$ is a section of $\pi_{XM}: M \to X$ defined on  $U$.

By varying the action $S$ with respect to admissible variations of $\psi=(\gamma, z)$, a critical point $\psi = (\gamma, z)$ of the Hamilton-Pontryagin action functional satisfies  
\begin{equation} \label{varS}
\begin{split}
\delta S(\psi)&=\frac{d}{d\lambda}\bigg|_{\lambda=0}S(\psi_{\lambda})\\
&=\frac{d}{d\lambda}\bigg|_{\lambda=0} \int_{U} \psi_{\lambda}^{\ast} ( \Theta_{M} -\mathcal{E})\\
&=-\int_{U} \psi^{\ast} (\mathcal{V}_{M} \! \text{\Large$\lrcorner$} \; \Omega_{M}) - \psi^{\ast}(\mathcal{V}_{M} \!\text{\Large$\lrcorner$} \; \mathbf{d}\mathcal{E})\\
&=-\int_{U} \psi^{\ast} (\mathcal{V}_{M} \! \text{\Large$\lrcorner$} \;\Omega_{\mathcal{E}})\\
&=0
\end{split}
\end{equation}
for all vertical vector fields $\mathcal{V}_{M}$ on $M$ that vanish on the boundary of $\psi(U)$ and such that
\begin{equation}\label{VertVectConst}
\psi^{\ast}(\mathcal{V}_{M}\! \text{\Large$\lrcorner$} \; \Phi)=0 \quad \mbox{for all} \quad \Phi \in F.
\end{equation}

Let $\psi: X \rightarrow M$ be a local section represented in local coordinates by $\psi(x) = (x^\mu, y^A(x), v^A_\mu(x), p_A^\mu(x), p(x))$.  The variation \eqref{varS} of the action integral is then locally given as 
\begin{equation*}
\begin{split}
\int_{U} \left\{ \delta p_{A}^{\mu} \left( \frac{\partial y^{A}}{\partial x^{\mu}}- v^{A}_{\mu} \right)+ 
 \left( -p_{A}^{\mu}+ \frac{\partial L}{\partial v^{A}_{\mu}} \right) \delta v^{A}_{\mu}+ 
 \left( -\frac{\partial p_{A}^{\mu}}{\partial x^{\mu}}+ \frac{\partial L}{\partial y^{A}}\right) \delta y^{A}
 \right\} d^{n+1}x=0
\end{split}
\end{equation*}
for all variations $\delta y^{A}$, $\delta v^{A}_{\mu}$ and $\delta p^{\mu}_{A}$, where the variation of $y^{A}$ vanishes at the boundary of $X$, namely, $\delta y^{A}\big|_{\partial X}=0$.   Imposing the condition that the variation be admissible then yields 
\begin{equation*}\label{LocalVertVectConst}
A^{\alpha}_{A} \delta{y}^{A}=0.
\end{equation*}
From this, we obtain the following equations of motion for {\bfi implicit Lagrange-d'Alembert equations}:
\begin{equation}\label{ImpLDE}
\begin{split}
\frac{\partial p_A^\mu}{\partial x^\mu} -  \frac{\partial L}{\partial y^A} =\lambda_{\alpha}A^{\alpha}_{A}, \quad
	\frac{\partial y^A}{\partial x^\mu} = v^A_\mu, \quad
	p_A^\mu = \frac{\partial L}{\partial v^A_\mu},
\end{split}
\end{equation}
where the $\lambda_\alpha$ are Lagrange multipliers, and
\begin{equation*} \label{energy-const}
\frac{\partial}{\partial x^\mu} (p + p_A^\mu v^A_\mu - L) = 0,
\end{equation*}
together with nonholonomic affine constraints
\[
A^{\alpha}_{A}v^{A}_{\mu}+A^{\alpha}_{\mu}=0,
\]
which serve to determine the $k$ Lagrange multipliers $\lambda^\alpha$.  Finally, by imposing the generalized energy constraint $E=0$ we may suppose that $p = L- p_A^\mu v^A_\mu$.  We summarize this in the following theorem.

\begin{theorem}
The following statements concerning a section $\psi$ of $\pi_{XM}: M \to X$ are equivalent:
\begin{itemize}
\item[(1)] $\psi$ is a critical point of  the action \eqref{HPAction} under admissible variations \eqref{admvariations};
\item[(2)] $\psi$ satisfies the implicit Lagrange-d'Alembert equations \eqref{ImpLDE};
\item[(3)] for all vertical vector fields $\mathcal{V}_{M}$ on $M$ along $\Delta_{M}$ such that $\psi^{\ast}(\mathcal{V}_{M}\! \text{\Large$\lrcorner$} \; \Phi)=0$ for all $\Phi \in F$, 
\begin{equation*}
\psi^{\ast}(\mathcal{V}_{M} \!\text{\Large$\lrcorner$} \;\Omega_{\mathcal{E}})=0.
\end{equation*}
\end{itemize}
\end{theorem}

\section{Examples}

In this section, we shall demonstrate some examples of Lagrange-Dirac field theories in the context of multi-Dirac structures. Namely, we show that the two examples of {\it scalar nonlinear wave equations} and {\it electromagnetism} can be described as standard Lagrange-Dirac field theories. Furthermore, as an example of the Lagrange-Dirac field theories with nonholonomic constraints, we consider {\it time-dependent mechanical systems with affine constraints}.

\paragraph{Example: Nonlinear Wave Equations.}

Let us consider the scalar nonlinear wave equation ({\bfi Klein-Gordon equations}) in $1 + 1D$, as discussed in \cite{Brid1997}.  The configuration bundle is $\pi_{XY}: Y \rightarrow X$, with $X = \mathbb{R}^2$ and $Y = X \times \times \mathbb{R}$.  Local coordinates are given by $(x^{0},x^{1})$ for $U \subset X$ and $(x^{0},x^{1},\phi)$ for $Y$, and hence $(x^{0},x^{1},\phi,v_{0},v_{1})$ for $J^{1}Y$ and $(x^{0},x^{1},\phi,p,p^{0},p^{1})$ for $Z \cong J^{1}Y^{\star}$.  The nonlinear wave equation is given by 
\[
\frac{\partial^{2} \phi }{{\partial x^{0}}^{2}} -\Delta \phi-V^{\prime}(\phi)=0,
\]
where $\phi$ is a section of $\pi_{XY}$ and $V : Y \rightarrow \mathbb{R}$ is a nonlinear potential.  The Lagrangian for this equation is given by
\[
L(x^{0},x^{1},\phi,v_{0},v_{1})=\left[ \frac{1}{2}(v_{0}^{2} -v_{1}^{2}) +V(\phi) \right]
\]
and the Lagrangian density is of course denoted by
\[
\mathcal{L}(x^{0},x^{1},\phi,v_{0},v_{1})=L(x^{0},x^{1},\phi,v_{0},v_{1})dx^{1} \wedge dx^{0}.
\]
Then, the generalized energy is given by
\[
E(x^{0},x^{1},\phi,v_{0},v_{1},p,p^{0},p^{1})=p+p^{0}v_{0}+p^{1}v_{1}-L(x^{0},x^{1},\phi,v_{0},v_{1}).
\]

\medskip
\noindent
\emph{(a)  Hamilton-Pontryagin Principle:}  The Hamilton-Pontryagin principle is given, in coordinates $(x^{0},x^{1},\phi,v_{0},v_{1},p,p^{0},p^{1})$ for $M=J^{1}Y \oplus Z$, by
\[
\begin{split}
\delta \int_{U} (p+p^{0}\phi_{,0}+p^{1}\phi_{,1}-E(x^{0},x^{1},\phi,v_{0},v_{1},p,p^{0},p^{1}))dt \wedge dx=0.
\end{split}
\]
with the condition of vanishing of $\delta{\phi}$ at the boundaries of $U$.
By direct computations, it follows that
\begin{equation}\label{ImpEulLagEq-NW}
\frac{\partial \phi}{\partial x^{0}}=v_{0}, \quad 
\frac{\partial \phi}{\partial x^{1}}=v_{1},  \quad p^{0}=v_{0},  \quad p^{1}=-v_{1},  \quad 
\frac{\partial p^0}{\partial x^{0}}
+
\frac{\partial p^1}{\partial x^{1}}
=V^{\prime}(\phi), 
\end{equation}
and imposing the generalized energy constraints $E=0$, one has
\begin{equation}\label{Gen-energy-const}
p=L(x^{0},x^{1},\phi,v_{0},v_{1})-p^{0}v_{0}-p^{1}v_{1}.
\end{equation}

\medskip
\noindent
\emph{(b)  Multi-Dirac Structures:}
On the other hand, since the canonical pre-multisymplectic 3-form on $M$ is given by
\[
\Omega_{M}=-dp_{0}\wedge d\phi \wedge dx^{1}+dp^{1}\wedge d\phi \wedge dx^{0}-dp \wedge dx^{0} \wedge dx^{1},
\]
one can define a multi-Dirac structure $D_M = D_1, D_2$ of degree $2$ associated to $\Omega_M$ on $M$ by the standard prescription.  We omit the definition of $D_1$, since as discussed previously, only $D_2$ is important for the formulation of the dynamics:
\[
D_{M,2}=\left\{ (\mathcal{X},\Sigma) \in P_{2}= T^{2}M \oplus \Lambda^{1}M \mid \mathbf{i}_{\mathcal{X}}\Omega_{M}=\Sigma \right\}.
\]
In this context, recall that the Lagrange-Dirac system for the field theories is given by a triple $(\mathcal{X}, E, D_{M})$ that satisfies
\[
(\mathcal{X},\mathbf{d}E) \in D_{M},
\]
where the 2-vector field $\mathcal{X}$ is given by
\[
\mathcal{X} 
	= \bigwedge_{\mu = 0}^{1}
		\left( \frac{\partial}{\partial x^\mu}
			+ C_\mu \frac{\partial}{\partial \phi}
			+ C_{\mu}^\nu \frac{\partial}{\partial p^\nu}			
			+ c_\mu \frac{\partial}{\partial p} \right).
\]
Hence, we have
\[
\Omega_{M}^{\flat} \cdot \mathcal{X}=- \mathbf{d}E.
\]
The direct computation using coordinates $(x^{0},x^{1},\phi,v_{0},v_{1},p,p^{0},p^{1})$ for $M$ leads to
\[
\begin{split}
&v_{0}=C_{0}, \qquad v_{1}=C_{1},\qquad C^{0}_{0}+C^{1}_{1}=\frac{\partial L}{\partial \phi}, \qquad p^{0}=\frac{\partial L}{\partial v_{0}},\qquad p^{1}=\frac{\partial L}{\partial v_{1}}\\
&c_{0}+C^{1}_{0}C_{1}-C_{0}C^{1}_{0}=\frac{\partial L}{\partial x}^{0}, \qquad c_{1}+C^{0}_{0}C_{1}-C_{0}C^{0}_{1}=\frac{\partial L}{\partial x^{1}},
\end{split}
\]
where
\[
C_{\mu}=\frac{\partial \phi}{\partial x^{\mu}}, \qquad C^{\nu}_{\mu}=\frac{\partial p^{\nu}}{\partial x^{\mu}}, \qquad c_{\mu}=\frac{\partial p}{\partial x^{\mu}}.
\]
Hence, by simple rearrangements, we can obtain implicit Euler-Lagrange equations \eqref{ImpEulLagEq-NW} for nonlinear waves, together with
\[
\frac{\partial p}{\partial x^{\mu}}=\frac{\partial }{\partial x^{\mu}} (L(x^{0},x^{1},\phi,v_{0},v_{1})-p^{0}v_{0}-p^{1}v_{1}).
\]
Needless to say, imposing the  generalized energy constraint $E=0$ yields equation \eqref{Gen-energy-const}.

\paragraph{Example: Electromagnetism.}

The multisymplectic description of electromagnetism can be found, among others, in \cite{GIMM1997}.  Here, we highlight the role of the Hamilton-Pontryagin variational principle and the associated multi-Dirac structure.  The electromagnetic potential $A = A_\mu dx^\mu$ is a section of the bundle $Y = T^\ast X$ of one-forms on spacetime $X$.  For the sake of simplicity, we let $X$ be $\mathbb{R}^4$ with the Minkowski metric, but curved spacetimes can be treated equally well.  The bundle $T$ has coordinates $(x^\mu, A_\mu)$ while $J^1 Y$ has coordinates $(x^\mu, A_\mu, A_{\mu, \nu})$.  The electromagnetic Lagrangian density is given by 
\[
	\mathcal{L}(A, j^1 A) = -\frac{1}{4} \mathbf{d} A \wedge \ast \mathbf{d} A, 
\]
where the Hodge star operator $\ast$ is associated to the Minkowski metric.

\medskip
\noindent
\emph{(a) Hamilton-Pontryagin Principle:}
The Hamilton-Pontryagin action principle is given in coordinates by 
\[
	S = \int_U \left[
	p^{\mu, \nu} \left( \frac{\partial A_\mu}{\partial x^\nu} - A_{\mu, \nu} \right)
	- \frac{1}{4} F_{\mu\nu} F^{\mu\nu}  \right] d^4x,
\]
where $U$ is an open subset of $X$ and $F_{\mu\nu}$ is the anti-symmetrization $F_{\mu\nu} = A_{\mu, \nu} - A_{\nu, \mu}$.  The Hamilton-Pontryagin equations follow as before:
\begin{equation} \label{HPmaxwell}
	p^{\mu,\nu} = F^{\mu\nu}, 
	\quad A_{\mu, \nu} = \frac{\partial A_\mu}{\partial x^\nu}, 
	\quad \frac{d p^{\mu, \nu}}{d x^\nu} = 0.
\end{equation}

\medskip
\noindent
\emph{(b)  Multi-Dirac Structure:} The canonical multisymplectic form for electromagnetism is given in coordinates by 
\[
	\Omega_M = d A_\mu \wedge d p^{\mu, \nu} \wedge d^n x_\nu 
		- dp \wedge d^{n + 1} x.
\]
Let $D = D_1, \ldots, D_{n + 1}$ be the multi-Dirac structure induced by this form and consider in particular the component $D_{n + 1}$ of the highest degree. In terms of this bundle, the implicit Euler-Lagrange equations can be described by 
\[
	(\mathcal{X}, (-1)^{n + 2}\mathbf{d} E_L) \in D_{n + 1}, 
\]
where $\mathcal{X}$ is a partial vector field of degree $n + 1$ given locally by 
\[
	\mathcal{X} = \bigwedge_{\mu = 1}^{n + 1} 
		\left(
			\frac{\partial}{\partial x^\mu}
			+ C_{\nu\mu} \frac{\partial}{\partial A_\nu}
			+ C_{\mu}^{\kappa\nu} \frac{\partial}{\partial 
				p^{\kappa,\nu}}			
			+ C_\mu \frac{\partial}{\partial p}
		\right) 
\]
and $E_L$ is the generalized energy density given by 
\[
	E_L(x^\mu, A_\mu, A_{\mu, \nu}, p^{\mu, \nu}, p) = p + p^{\mu, \nu} A_{\mu, \nu} 
		+ \frac{1}{4} F_{\mu\nu} F^{\mu\nu}.
\]
As before, a coordinate computation then shows that the coefficients of $\mathcal{X}$ are given by 
\[
	C_{\mu \nu} = A_{\mu, \nu},
	\quad
	C_\nu^{\mu\nu} = \frac{\partial L}{\partial A_\mu} = 0,
	\quad 
	p^{\mu,\nu} = \frac{\partial L}{\partial A_{\mu, \nu}} = 
	F^{\mu\nu}
\]
and 
\[
	C_\mu + C_\mu^{\kappa \lambda} C_{\kappa \lambda} 
		- C_{\nu\mu}C^{\nu\lambda}_\lambda =
	\frac{\partial L}{\partial x^\mu} = 0,
\]
where 
\[
	C_{\nu\mu} = \frac{\partial A_{\nu}}{\partial x^\mu}, 
	\quad 
	C^{\kappa\lambda}_\mu = 
		\frac{\partial p^{\kappa, \lambda} }{\partial x^\mu}, 
	\quad \text{and} \quad
	C_\mu = \frac{\partial p}{\partial x^\mu}.
\]
It is easy to show that these equations are equivalent to Maxwell's equations \eqref{HPmaxwell} in implicit form, together with the energy constraint \eqref{encons-LDS}.

\paragraph{Example: Time-Dependent Mechanical Systems with Affine Constraints.}
Let $\pi_{XY}:Y \to X$ be a finite dimensional fiber bundle called a {\it covariant configuration bundle} over a oriented manifold $X$.
In the case of a time-dependent mechanical system, we can set 
$X=\mathbb{R}$ and $Y=\mathbb{R} \times Q$, where $Q$ is an $n$-dimensional configuration manifold. Namely, $Y$ is an $(n+1)$-dimensional differentiable manifold, with local coordinates $(t, q^{i})$, where $t \in X$ indicates the time and $q^{i}$ are fiber coordinates of $Y$. The first jet bundle of local sections of $\pi_{XY}$ is given  by $J^{1}Y \cong \mathbb{R} \times TQ$, which is denoted, in local coordinates $(t,q^{i},v^{i})$ for $\gamma \in J^{1}Y$, by $\gamma=dt \otimes \left( {\partial}/{\partial t}+ v^{i} {\partial}/{\partial q^{i}}\right)$.
Let $\mathcal{L}:J^{1}Y \to \Lambda^{1}X$ be a Lagrangian density, which is denoted by 
\[
\mathcal{L(\gamma)}=L(t,q^{i},v^{i})dt,
\]
where $L$ is a Lagrangian, possible degenerate, which is defined on $J^{1}Y$ and $\eta=dt$ is the volume form on $\mathbb{R}$. The dual jet bundle $Z(\cong J^{1}Y^{\star})= T^{\ast}Y \cong T^{\ast}\mathbb{R} \times T^{\ast}Q$ is the vector bundle over $Y$, with local coordinates $(t,q^{i},p_{t},p^{i})$. The canonical one-form on $Z$ is given by $\Theta=p_{i}dq^{i} +p_{t}dt$ and  the canonical two-form is given by $\Omega=-\mathbf{d}\theta=dq^{i} \wedge dp_{i}-dp_{t} \wedge dt$. 
\medskip

Let us consider a time-dependent mechanical system which is constrained by affine nonholonomic constraints, where $\Delta_{Y}$ is a nontrivial distribution on $Y=\mathbb{R} \times Q$, which is given by, for each $(t,q) \in Y$,
\begin{equation*}\label{ConstDist}
\Delta_{Y}(t,q)=\left\{ j^{1}\phi(t) \in J^{1}Y_{(t,q)} \mid \left< \varphi_{r}(t,q), j^{1}\phi(t) \right> =0  \right\},
\end{equation*}
where $\phi_{r}$ are the $k$ one-forms given by
\[
\varphi_{r}=A_{i}^{r}(t,q)dq^{i} +B^{r}(t,q)dt, \quad r=1,...,k.
\]
 Then, the distribution $\Delta_{Y}$ induces an affine subbundle $\mathcal{C}$ of $J^{1}Y =\mathbb{R} \times TQ$ locally given by
\[
\mathcal{C}=\{ (t,q^{i},v^{i}) \in J^{1}Y \mid A_{i}^{r}(t,q)v^{i} +B^{r}(t,q)=0 \}.
\]

Let $M=J^{1}Y \oplus Z=(\mathbb{R} \times TQ) \oplus (T^{\ast}\mathbb{R} \times T^{\ast}Q)$ be the Pontryagin bundle over $Y=\mathbb{R} \times Q \to X=\mathbb{R}$ and let $\pi_{YM}: M \to Y$ be a canonical projection given by $(t,q^{i}, v^{i},p_{t},p_{i}) \mapsto (t,q^{i})$ and define the distribution on $M$ by
\begin{equation}\label{DistZ}
\Delta_{M}=(T\pi_{YM})^{-1}|_{J^{1}Y}(\mathcal{C}) \subset TM
\end{equation}
and its annihilator $\Delta_{M}^{\circ}$ is spanned by  $k$ {\it horizontal} one-forms on $M$
\begin{equation*}\label{OneFormZ}
\tilde{\varphi}_{r}=A_{i}^{r}(t,q)(dq^{i} -v^{i}dt), \quad r=1,...,k.
\end{equation*}
 
Let $P=TM \oplus \Lambda^{1}M$ be the graded Pontryagin bundle and $\left<\left< (\mathcal{X},  \Sigma), (\bar{\mathcal{X}}, \bar{\Sigma}) \right>\right>_{-} $ in \eqref{gsym} exactly corresponds to
\[
\left<\left< (\mathcal{X},  \Sigma), (\bar{\mathcal{X}}, \bar{\Sigma}) \right>\right>_{-}=\frac{1}{2}\left( \mathbf{i}_{\bar{\mathcal{X}}}\Sigma+\mathbf{i}_{\mathcal{X}}\bar{\Sigma} \right)
\]
for $(\mathcal{X}, \Sigma) \in P$ and $(\bar{\mathcal{X}}, \bar{\Sigma}) \in P$. 

Let $\pi_{ZM}: M \to Z$ and the induced (multi-)Dirac structure $D_{M}:=D_{M,1} \subset TM \oplus T^{\ast}M$ on $M$ induced from the affine constraint distribution $\Delta_{M}$ in \eqref{DistZ} for time-dependent mechanics may be given for the case in which $n=0$ and $r=s=1$. This is given by
\begin{equation}\label{DiracField}
\begin{split}
D_{M}=\left\{ (\mathcal{X}, \Sigma) \in TM \oplus T^{\ast}M  \mid  \; \Sigma-\mathbf{i}_{\mathcal{X}} \Omega_{M} \in \Delta_{M}^{\circ} \quad  \textrm{for} \quad  
\mathcal{X} \in \Delta_{M} \right\},
\end{split}
\end{equation}
where $\Omega_{M}:=\pi_{ZM}^{\ast}\Omega$. 

Now the partial vector field $\bar{\mathcal{X}}: M \to TZ$ is denoted in local coordinates $(t,q^{i},v^{i},p_{i},p_{t})$ on $M=J^{1}Y \oplus Z$ as
\begin{equation*} \label{multvf}
	\bar{\mathcal{X}} =
		\frac{\partial}{\partial t}
			+ \dot{q}^i \frac{\partial}{\partial q^{i}}
			+ \dot{p}_{i} \frac{\partial}{\partial p_i}
			+ \dot{p_{t}} \frac{\partial}{\partial p_{t}}, \qquad
         \mathbf{i}_{\mathcal{X}}(\pi_{XZ}^{\ast}\eta)=1,
\end{equation*}
where $\eta=dt$ is the volume form on $X=\mathbb{R}$. Then, the vertical lift of the partial vector field $\mathcal{X}: M \to TM$ is given by
$\bar{\mathcal{X}}=T\pi_{XM}(\mathcal{X})$,
where $T\pi_{XM}: TM \to TZ$ is the bundle map associated to the canonical projection $\pi_{XM}: M \to Z$. In coordinates, 
one has
\begin{equation} \label{MultVectField-time}
	\mathcal{X} =
		\frac{\partial}{\partial t}
			+ \dot{q}^i \frac{\partial}{\partial q^{i}}
			+ \dot{p}_{i} \frac{\partial}{\partial p_i}
			+ \dot{p_{t}} \frac{\partial}{\partial p_{t}}, \qquad
         \mathbf{i}_{\mathcal{X}}(\pi_{XM}^{\ast}\eta)=1,
\end{equation}
where we note that the horizontal components of $v^{i}$ are zero.
\medskip

On the other hand, the generalized energy $E: M \to \Lambda^{1}X=\mathbb{R}$ is given by
\begin{equation*}\label{GenEnergy-time}
\begin{split}
E(t,q^{i},v^{i},p_{t},p_{i})=p_{t}+p_{i}v^{i}-L(t,q^{i},v^{i})
\end{split}
\end{equation*}
and the differential of $E$ is locally given by 
\begin{equation}\label{DiffGenEnergy-time}
\begin{split}
\mathbf{d}E(t,q^{i},v^{i},p_{t},p_{i})=\left( -\frac{\partial{L}}{\partial{t}}\right)dt+\left( -\frac{\partial{L}}{\partial{q^{i}}}\right)dq^{i}+\left(p_{i}-\frac{\partial{L}}{\partial{v^{i}}}\right) dv^{i}+dp_{t}+v^{i}dp_{i}.
\end{split}
\end{equation}

Thus, the {\it nonholonomic Lagrange-Dirac system for time-dependent mechanics} is given by a triple $(\mathcal{X},E, D)$ which satisfies the condition
\begin{equation*} 
(\mathcal{X}, \mathbf{d}E) \in D,
\end{equation*}
which yields the intrinsic nonholonomic Lagrange-Dirac systems as 
\[
\mathbf{i}_{\mathcal{X}}\Omega-\mathbf{d}E \in \Delta_{M}^{\circ}.
\] 
In coordinates, one obtains
\begin{equation*}
\begin{split}
\dot{t}=1,\quad \dot{q}^{i}=v^{i}, \quad \dot{p_{i}}-\frac{\partial L}{\partial q^{i}}=\lambda_{r}A^{r}_{i},
\quad p_{i}-\frac{\partial L}{\partial v^{i}}=0,
\end{split}
\end{equation*}
together with the affine constraints
\[
A^{r}_{i}(t,q)v^{i}+B^{i}(t,q)=0
\]
and the energy conservation
\[
\frac{d}{dt}\left( p_{t}+p_{i}v^{i}-L(t,q^{i},v^{i}) \right)=0.
\]
Imposing the generalized energy constraint $E=0$, one can recover
\[
p_{t}=L-p_{i}v^{i}.
\] 

\section{Conclusions and Future Work}

In this paper, we have developed a notion of multi-Dirac structure on the Pontryagin bundle and its associated Lagrange-Dirac field theories, where we allow a given Lagrangian to be degenerate. We have shown that the Hamilton-Pontryagin variational principle for field theories induces implicit Euler-Lagrange equations for fields, which is consistent with the standard case of the Lagrange-Dirac field theory. The ideas are shown to be a natural extension of Dirac structures and Lagrange-Dirac dynamical systems in mechanics. We have further developed the induced multi-Dirac structure from distributions and with the associated nonholonomic Lagrange-Dirac field theories. Additionally, we have also shown that the implicit Lagrange-d'Alembert field equations, which are equivalent with the nonholonomic Lagrange-Dirac systems, can be formulated by the Lagrange-d'Alembert-Pontryagin principle. Finally, we have demonstrated the Lagrange-Dirac field theories by illustrative examples of scalar nonlinear waves (Klein-Gordon equations), electromagnetism and time-depedent mechanical systems with affine constraints.

\paragraph{Poisson Brackets for Field Theories.}  In the literature on classical field theories, a number of different Poisson brackets have been proposed.  We have already discussed the multi-Poisson bracket of \cite{CIdeLe96}, but other noteworthy developments include the spacetime covariant bracket of \cite{MaMoMo1986} and the work of \cite{FoRo2005} on relating the multisymplectic framework to the Peierls-de Witt bracket.  We hope to shed further light on these various brackets using multi-Dirac structures.

\paragraph{Discretizations of Multi-Dirac Structures and Field Theories.}

\cite{LeOh2008} recently proposed a concept of \emph{discrete Dirac structure} in mechanics.  At the same time, it has been shown by \cite{BoRaMa2009} that discrete versions of the Hamilton-Pontryagin principle can be used to derive very accurate geometric integrators for mechanical systems.  We propose to explore the field-theoretic version of these discretizations: since multi-Dirac structures and variational principles are a natural way to describe field theories even for degenerate Lagrangians, this approach is especially promising.  Moreover, in related work (see \cite{CVGMY2010}) we show that the field-theoretic Hamilton-Pontryagin principle incorporates the \emph{adjoint formalism} for covariant field theories.  By constructing discrete counterparts of these results, we hope to come up with a discrete representation of field theories which respects the gauge freedom of the underlying theory.

\paragraph{Space-Time Decomposition of Multi-Dirac Structures.}

The multi-Dirac structures introduced in this paper are \emph{covariant} in the sense that no distinction is made between time and the spatial variables on the base space.  In fact, as we have pointed out before, the configuration bundle $\pi_{XY}: Y \to X$ can be arbitrary and in particular $X$ does not have to represent spacetime.

While this approach has various advantages (see \cite{GolSter73, KiTu1979, Brid1997, GIMM1997} for a discussion), not in the least that the resulting structures are all finite-dimensional, it is often necessary to single out one coordinate direction as time, for instance when considering the initial-value problem for field theories (see \cite{Wald1984}).  This procedure is referred to as choosing a \emph{(bundle) slicing}.  A comprehensive overview of the theory of slicings and the corresponding instantaneous formulations of field theories can be found in \cite{GIMM1999, BiSnFi1988}.

We intend to investigate the behavior of multi-Dirac structures in the presence of a bundle slicing.  Given the choice of a slicing, we expect that a multi-Dirac structure will induce an infinite-dimensional Dirac structure on the instantaneous space of fields.  This is similar to the results of \cite{GIMM1999}, who show that the multisymplectic form on the dual of a jet bundle induces a (weakly) symplectic form on the space of fields once a slicing has been chosen.  Among other things, we expect that establishing the correspondence between the theory of covariant, multi-Dirac structures and their instantaneous counterparts will be important in a number of areas, most notably the theory of Poisson brackets for field theories.  We refer to \cite{MaMoMo1986} and \cite{FoRo2005} for two different approaches to the theory of covariant Poisson brackets for field theories.  

\paragraph{Symmetry and Reduction of Multi-Dirac Structures.}  

The presence of non-trivial symmetry groups is an essential feature of many classical field theories, such as electromagnetism or Yang-Mills theories.  More complex examples are the theory of perfect fluids and Einstein's theory of relativity.  It is often useful to obtain a form of the field equations from which the symmetry has been factored out; this is for instance what happens in writing Maxwell's equations in terms of the field strength $F_{\mu\nu}$ rather than the vector potential $A_\mu$.  From a bundle-theoretic point of view, various approaches of symmetry reduction for field theories have been proposed (see \cite{CLM2003, CaRa2003, GaRa2010} and the references therein).  We propose to develop a reduction theory for multi-Dirac structures along the lines of the theory of symmetry reduction for  Dirac structures in classical mechanics developed by \cite{BR2004, YoMa2007, YoMa2009}.  In this way, we expect to unify and extend previous results.

\paragraph{The Relation with Stokes-Dirac Structures.}

 In the introduction, we mentioned the concept of \emph{Stokes-Dirac structures}, which are infinite-dimensional structures based on the exterior differential and Stokes theorem that satisfy the axioms necessary for a Dirac structure.  These structures were introduced by \cite{VaMa2002} as a means to describe boundary control for field theories in an intrinsic way.  In \cite{VaYoMa2010}, it was shown that Stokes-Dirac structures appear through symmetry reduction for a certain canonical infinite-dimensional Dirac structure.

Given the relevance of Stokes-Dirac structures for the control of field theories, it is important to investigate the link with the multi-Dirac structures introduced in this paper.   Based on the discussion of slicings in the context of multi-Dirac structures and the results of \cite{VaYoMa2010}, we may conclude that a Stokes-Dirac structure can be obtained starting from a multi-Dirac structure by first choosing a slicing and then performing (instantaneous) symmetry reduction.  One can now ask whether these two operations can be reversed: do we obtain the same Stokes-Dirac structure by first performing (covariant) multi-Dirac reduction and then choosing a slicing?  In other words, when does the following diagram commute?

\[\xymatrix{
	\text{Canonical Multidirac structure} \ar[d]_{\text{Covariant reduction}} 
		\ar[rr]^{3 + 1} & &
		\text{Infinite-Dimensional Dirac structure}
	 \ar[d]^{\text{Instantaneous reduction}} \\
	\text{Multidirac structure} \ar[rr]_{3 + 1} & &
		\text{Stokes-Dirac structure}
}\]

\appendix
\section{Generalities about Multi-Vector Fields}

\subsection{Definitions}

Let $M$ be a manifold of dimension $n$.  A $k$-multivector field (where $k \le n$) on $M$ is a section $\mathcal{X}$ of the $k$-fold exterior power $T^k M$ of the tangent bundle.  The module of all $k$-multivector fields is denoted by $\mathfrak{X}^k(M)$.

We define the {\bfi interior product} of multivector fields and forms as follows.  Let $\mathcal{X} \in \mathfrak{X}^k(M)$ and $\alpha \in \Omega^l(M)$, where $k \le l$.  The {\bfi left interior product} $\mathcal{X} \intprod \alpha$ is then the unique $(l-k)$-form such that
\[
	(\mathcal{X} \intprod \alpha)(\mathcal{X}') = \alpha(\mathcal{X} \wedge \mathcal{X}')
\]
for all $(l - k)$-multivectors $\mathcal{X}'$.  Similarly, the {\bfi right interior product} of a multivector $\mathcal{X} \in \mathfrak{X}^k(M)$ and a form $\beta \in \Omega^m(M)$, with $k \ge m$, is the unique $(k-m)$-multivector $\mathcal{X} \prodint \beta$ with the property that 
\[
	(\mathcal{X} \prodint \beta) \intprod \gamma = (\beta \wedge \gamma)(\mathcal{X})
\]
for all $(k-m)$-forms $\gamma$ on $M$.

Further information about multivectors and interior products can be found in (for instance) \cite{Tu1974, Marle1997}.  From these papers, we quote the following two results on the relation between the interior product and the wedge product of vector fields and forms, respectively.

\begin{lemma}[\cite{Tu1974}] \label{lemma:wedge}
Let $\mathcal{X} \in \mathfrak{X}^k(M)$ and $\alpha \in \Omega^1(M)$.  Then 
\[
	(\mathcal{X} \wedge \mathcal{X}') \prodint \alpha = 
		(\mathcal{X} \prodint \alpha) \wedge \mathcal{X}' + 
		(-1)^k \mathcal{X} \wedge (\mathcal{X}' \prodint \alpha)
\]
for any multivector $\mathcal{X}'$, and 
\[
	\mathcal{X} \intprod (\alpha \wedge \beta) = (\mathcal{X} \prodint \alpha) \intprod \beta
		+ (-1)^k \alpha \wedge (\mathcal{X} \intprod \beta)
\]
for any form $\beta$.
\end{lemma}

When the multi-vector field is decomposable, the right interior product with a one-form can easily be computed by means of the following formula, the proof of which follows from the previous theorem.

\begin{lemma} \label{lemma:decomp}
	Let $\mathcal{X}$ be a $k$-multi-vector field which is {\bfi decomposable} in the sense that $\mathcal{X} = \bigwedge_{\mu = 1}^k \mathcal{X}_\mu$, where $\mathcal{X}_\mu$ ($\mu = 1, \ldots, k$) are vector fields and consider a one-form $\alpha$.  Then
	\[
		\mathcal{X} \prodint \alpha = 
			\sum_{\mu = 1}^k (-1)^{\mu + 1} \left<\mathcal{X}_\mu, \alpha \right> 
				\hat{\mathcal{X}}_\mu, 
	\]
	where $\hat{\mathcal{X}}_\mu$ is the $(k-1)$-vector field obtained by deleting $\mathcal{X}_\mu$, i.e. 
	\[
		\hat{\mathcal{X}}_\mu = \bigwedge_{\stackrel{\lambda = 1}{\lambda \ne \mu}}^{n+1} \mathcal{X}_\lambda.
	\]
\end{lemma}

\begin{corollary} \label{cor:prodprod}
	Let $\mathcal{X}$ be a decomposable $k$-multivector and $\alpha$ an $l$-form, where $l \ge k$.  Then $\mathcal{X} \prodint (\mathcal{X} \intprod \alpha ) = 0$.
\end{corollary}
\begin{proof}
Let $\beta = \mathcal{X} \intprod \alpha$.  With the notations of the previous lemma, we have that 
\[
\mathcal{X} \prodint \beta = 
			\sum_{\mu = 1}^k (-1)^{\mu + 1} \left<\mathcal{X}_\mu, \beta \right> 
				\hat{\mathcal{X}}_\mu, 
	\]
but $\left<\mathcal{X}_\mu, \beta \right> = 0$ for all $\mu = 1, \ldots, k$, since
$\left<\mathcal{X}_\mu, \beta \right> = 
	\mathbf{i}_{\mathcal{X}_\mu} \mathbf{i}_{\bigwedge_{\lambda = 1}^{n+1} \mathcal{X}_\lambda} \alpha = 0$.
\end{proof}

In this paper, we also let $\hat{\mathcal{X}}_{\mu\nu}$ (where $\mu \ne \nu$) be the $(k-2)$-multivector obtained by deleting both $\mathcal{X}_\mu$ and $\mathcal{X}_\nu$: 
\[
	\hat{\mathcal{X}}_{\mu\nu} = \bigwedge_{\stackrel{\lambda = 1}{\lambda \ne \mu, \nu}}^{n+1} \mathcal{X}_\lambda.
\]
If $\mu < \nu$, we then have that 
\[
	\mathcal{X}_\lambda \wedge \hat{\mathcal{X}}_{\mu\nu} 
		= (-1)^{\mu+1} \delta_{\mu\lambda} \hat{\mathcal{X}}_\nu 
		-  (-1)^{\nu+1} \delta_{\nu\lambda} \hat{\mathcal{X}}_\mu.
\]
If $\mu > \nu$, note that $\hat{\mathcal{X}}_{\mu\nu} = \hat{\mathcal{X}}_{\nu\mu}$.  The above formula can then be applied.

\newtext{
The following lemmas are used in the proof of proposition~\ref{prop: dirac-const}.  In both cases, let $\Delta$ be a distribution on a manifold $M$ and denote its annihilator by $\Delta^\circ \subset T^\ast M$.   We use the following notations for the exterior products: 
\[
	\mbox{$\bigwedge$}^k \Delta = \underbrace{\Delta \wedge \cdots \wedge \Delta}_{\text{$k$ times}} \quad 
	\text{and} \quad
\mbox{$\bigwedge$}^l (\Delta^\circ) = \underbrace{\Delta^\circ \wedge \cdots \wedge \Delta^\circ}_{\text{$l$ times}},	 
\]
for all $k, l$.  Note especially that $\mbox{$\bigwedge$}^l (\Delta^\circ)$ is not equal to $(\mbox{$\bigwedge$}^l \Delta)^\circ$, since the latter is given by 
\[
	(\mbox{$\bigwedge$}^l \Delta)^\circ 
	= \Delta^\circ \wedge \Lambda^{l - 1}(M).
\]

\begin{lemma}  \label{lemma:vectorfield}
	Consider an $l$-vector field $\mathcal{X}$ and let $k$ be a fixed integer, $k \ge l$.  If $\mathbf{i}_{\mathcal{X}} \alpha = 0$ for all $\alpha \in \bigwedge^k \Delta^\circ$, then 
	\begin{equation} \label{decompose}
		\mathcal{X} \in \Delta \wedge T^{l - 1} M.
	\end{equation}
\end{lemma}
\begin{proof}
The proof proceeds by induction on the degree $l$ of $\mathcal{X}$.  When $\mathcal{X}$ is an ordinary vector field, the hypothesis is trivially true.  For multivector fields $\mathcal{X}$ of higher degree, we focus first on the case where $\mathcal{X}$ can be written as $X \wedge \mathcal{Y}$ (the general case then follows by linearity).  We have that 
\[
	0 = \mathbf{i}_{\mathcal{X}} \alpha =  \mathbf{i}_{\mathcal{Y}} \left(\mathbf{i}_X \alpha \right),
\]
for all  $\alpha \in \bigwedge^k \Delta^\circ$.  Since $\mathbf{i}_X \alpha \in \bigwedge^{k - 1} \Delta^\circ$, this implies by induction that $\mathcal{Y} \in \Delta \wedge T^{l - 2} M$.  In turn, this implies that $\mathcal{X} = X \wedge \mathcal{Y}$ satisfies \eqref{decompose}.
\end{proof}

\begin{lemma} \label{lemma:form}
	Consider a $k$-form $\alpha$ and let $l$ be a fixed integer, $l \le k$.  If $\mathbf{i}_{\mathcal{X}} \alpha = 0$ for all $\mathcal{X} \in \Delta \wedge T^{l - 1} M$, then 
	\[
		\alpha \in \underbrace{\Delta^\circ \wedge \cdots \wedge \Delta^\circ}_{\text{$k$ times}}.
	\]
\end{lemma}
\begin{proof}
The proof proceeds similarly by induction on the degree $k$ of $\alpha$.
\end{proof}
}

\subsection{Proof of \eqref{contract}}
\label{sec:comp}

The computation of equation \eqref{contract} is somewhat involved and depends on a number of coordinate identities, listed here.  First, recall that $\mathcal{X}$ is an $(n + 1)$-multivector field with local expression \eqref{multvf}, and that the multi-symplectic form $\Omega$ is locally given by equation \eqref{msform}.

Recall the volume form $\eta$ on $X$ is locally expressed by $d^{n + 1} x$ and denote
\[
	d^n x_\mu := \frac{\partial}{\partial x^\mu} \intprod d^{n+1}x, \quad
	d^{n-1}x_{\mu\nu} := \frac{\partial}{\partial x^\nu} \intprod d^n x_\mu, \quad \ldots
\]	
The following contractions will be useful:
\begin{equation} \label{contr1}
	\hat{\mathcal{X}}_\mu \intprod d^n x_\nu  = (-1)^{\nu + n} \delta_{\mu\nu}
\end{equation}
and, for $\mu < \nu$,
\begin{equation} \label{contr2}
	\hat{\mathcal{X}}_{\mu\nu} \intprod d^n x_\lambda = 
		(-1)^{n + \mu + \nu} ( \delta_{\nu\lambda} dx^\mu + \delta_{\mu\lambda} dx^\nu).
\end{equation}

Now return to the calculation of equation \eqref{contract} and the contraction of $\mathcal{X}$ with $\Omega_{M}$ is now given by 
\[
	\mathcal{X} \intprod \Omega_{M} = 
		\mathcal{X} \intprod (dy^A \wedge dp_A^\mu \wedge d^n x_\mu) -
		\mathcal{X} \intprod (dp \wedge d^{n+1} x), 
\]
and the two terms on the right-hand side will be calculated separately.

We begin with the first term. Using lemma~\ref{lemma:wedge}, we have 
\[
\mathcal{X} \intprod (dy^A \wedge dp_A^\mu \wedge d^n x_\mu)
= (\mathcal{X} \prodint dy^A) \intprod (dp_A^\mu \wedge d^n x_\mu) + 
(-1)^{n+1} dy^A \wedge ( \mathcal{X} \intprod (dp_A^\mu \wedge d^n x_\mu)).
\]
Both terms can be calculated using lemma~\ref{lemma:decomp} and using \eqref{contr1} and \eqref{contr2}.  After some rearrangements, the result is that 
\[
	\mathcal{X} \intprod (dy^A \wedge dp_A^\mu \wedge d^n x_\mu) = 
	(-1)^{n+2} \left[ (C^A_\mu C_{A\lambda}^\lambda - C^A_\lambda C^\lambda_{A\mu}) dx^\mu + C^A_\mu dp_A^\mu \right].
\]

For the second term, we use the same techniques to conclude that 
\[
	\mathcal{X} \intprod (dp \wedge d^{n+1} x) = (-1)^{n+1} \left[ dp - C_\mu dx^\mu \right].
\]
Putting the results for both terms together, we obtain \eqref{contract}.

\end{document}